\documentclass[11pt, letterpaper]{article}

\pdfoutput=1

\usepackage{geometry}
\usepackage{amsmath,amsfonts,amssymb,amsthm,booktabs,color,epsfig,graphicx,hyperref,url}

\usepackage[utf8]{inputenc} % allow utf-8 input
\usepackage[english]{babel}
\usepackage[T1]{fontenc}    % use 8-bit T1 fonts
\usepackage{hyperref}       % hyperlinks
\usepackage{url}            % simple URL typesetting
\usepackage{booktabs}       % professional-quality tables
\usepackage{amsfonts}       % blackboard math symbols
\usepackage{nicefrac}       % compact symbols for 1/2, etc.
\usepackage{microtype}      % microtypography
\usepackage{lipsum}
\usepackage{comment}
\usepackage{amsmath}
\usepackage{mathtools}
\usepackage{graphicx}       %package to manage images
\graphicspath{ {images/} }
\usepackage{parskip}
\usepackage{amsmath,amssymb}
\usepackage{comment}
\usepackage{xcolor}
\usepackage{amssymb}
\usepackage{amsthm}
\usepackage{setspace}
\usepackage[title]{appendix}

\usepackage{algorithm}% http://ctan.org/pkg/algorithms
\usepackage{algpseudocode}% http://ctan.org/pkg/algorithmicx
\makeatletter
\renewcommand{\ALG@name}{Procedure}
\makeatother
\usepackage{optidef}

\theoremstyle{plain}% Theorem-like structures provided by amsthm.sty
\newtheorem{theorem}{Theorem}

\newtheorem{lemma}{Lemma}

\theoremstyle{definition}

\newcommand{\vertiii}[1]{{\left\vert\kern-0.25ex\left\vert\kern-0.25ex\left\vert #1 
    \right\vert\kern-0.25ex\right\vert\kern-0.25ex\right\vert}}

\definecolor{DarkBlue}{rgb}{0,.08,.45}

\newcommand{\indep}{\rotatebox[origin=c]{90}{$\models$}}
\newcommand{\Sigmab}{\boldsymbol{\Sigma}}
\newcommand{\Omegab}{\boldsymbol{\Omega}}
\newcommand{\Lb}{\boldsymbol{L}}
\newcommand{\thetab}{\boldsymbol{\theta}}
\newcommand{\vb}{\boldsymbol{v}}
\newcommand{\Vb}{\boldsymbol{V}}
\newcommand{\Mb}{\boldsymbol{M}}

\newcommand{\Db}{\boldsymbol{D}}
\newcommand{\Xb}{\boldsymbol{X}}

\newcommand{\Zb}{\boldsymbol{Z}}

\newcommand{\Ab}{\boldsymbol{A}}
\newcommand{\eb}{\boldsymbol{e}}

\newcommand{\Ib}{\boldsymbol{I}}

\newcommand{\tT}{\text{T}}

\newcommand{\Lomega}{\boldsymbol{L}_{\boldsymbol{\Omega}}}
\newcommand{\Lsigma}{\Lb_{\Sigmab}}
\newcommand{\new}{\textcolor{black}}
\newcommand{\Cis}{C_i \text{'s} > 1}
\newcommand{\Sigmaobs}{\Sigmab_{obs}}
\newcommand{\lnn}{\text{ln}}

\newcommand{\newnew}{\textcolor{black}}

\newcommand{\newup}{\textcolor{black}}

\usepackage{natbib}
\setcitestyle{authoryear}

\title{Learning Gaussian graphical models with latent confounders}
\author{Ke Wang\footnote{kewang01@ucsb.edu} \quad Alexander Franks \quad Sang-Yun Oh \\
\\
\small{Department of Statistics and Applied Probability, University of California Santa Barbara}}

\date{}

\begin{document}

\maketitle

\begin{abstract}
%% Text of abstract
Gaussian Graphical models (GGM) are widely used to estimate network structure in domains ranging from biology to finance. In practice, data is often corrupted by latent confounders which biases inference of the underlying true graphical structure.  In this paper, we compare and contrast two strategies for inference in graphical models with latent confounders: Gaussian graphical models with latent variables (LVGGM) and PCA-based removal of confounding (PCA+GGM). While these two approaches have similar goals, they are motivated by different assumptions about confounding.  In this paper, we explore the connection between these two approaches and propose a new method, which combines the strengths of these two approaches. We prove the consistency and convergence rate for the PCA-based method and use these results to provide guidance about when to use each method. We demonstrate the effectiveness of our methodology using both simulations and two real-world applications.

\end{abstract}

%%Graphical abstract
% \begin{graphicalabstract}
%\includegraphics{grabs}
% \end{graphicalabstract}

% keywords can be removed
%\begin{keyword}
%Gaussian graphical models \sep High-dimensional data \sep Latent variables \sep Principal component analysis     
%\MSC[2020] Primary 62H22 \sep Secondary 62H25
%\end{keyword}

%\end{frontmatter}

\section{Introduction} \label{sec:introduction}
In many domains, it is useful to characterize relationships between features using network models.  For example, networks have been used to identify transcriptional patterns and regulatory relationships in genetic networks and applied as a way to characterize functional brain connectivity and cognitive disorders \cite{fox2007spontaneous,van2012influence,barch2013function,price2014multiple}.  One of the most common methods for inferring a network from observations is the Gaussian graphical model (GGM).  A GGM is defined with respect to a graph, in which the nodes correspond to joint Gaussian random variables and the edges correspond to the conditional dependencies among pairs of variables. A key property of the GGM is that the presence or absence of edges can be obtained from the precision matrix for multivariate Gaussian random variables  \cite{lauritzen1996graphical}. Similar to LASSO regression \cite{tibshirani1996regression}, we can infer the \new{sparse} graph structure via sparse precision matrix estimation with $l_1$-regularized maximum likelihood estimation.  This family of approaches is called graphical lasso (Glasso) \cite{friedman2008sparse,yuan2007model}. 

In practice, however, network inference may be complicated due to the presence of latent confounders. For example, when characterizing relationships between the stock prices of publicly trade companies, the existence of overall market and sector factors induces extra correlation between stocks \cite{choi2011learning}, which can obscure the underlying network structure between companies. 

\new{We focus on estimating $\Omegab = \Sigmab^{-1}$, the precision matrix encoding the graph structure of interest  \cite{yuan2007model, friedman2008sparse,cai2011constrained, parsana2019addressing}. }When latent confounders are present, the covariance matrix for the observed data, ${\Sigmab}_{obs}$ can be expressed as 
\begin{equation}
    \begin{aligned}
    {\Sigmab}_{obs}=\Sigmab+\Lb_{\Sigmab} , \label{eqn:splusl_sigma}
    \end{aligned}
\end{equation}
where the positive semidefinite matrix $\Lb_{\Sigmab}$ reflects the effect of latent confounders.
\newnew{One approach, which we will call PCA+GGM, is motivated by} confounders that affect the marginal correlation between observed variables \cite{parsana2019addressing} and uses principal component analysis (PCA) as a preprocessing step to remove the effect of these confounders \new{\cite{jolliffe2005principal, bartholomew2011latent}}. PCA removes the leading eigencomponents from ${\Sigmab}_{obs}$ which are assumed to be $\Lb_{\Sigmab}$, then a second stage of standard GGM inference follows.  %We call this PCA-based approach PCA+GGM. %For this approach to be effective, the norm of $\Lb_{\Sigmab}$ (noise) must be large relative to the norm\footnote{\new{As illustrated in later discussion, the norm here refers to the matrix spectral norm.}} of $\Sigmab$ (signal). 
PCA+GGM has shown to be useful in  estimating gene co-expression networks, where correlated measurement noise and batch effects induce large extraneous marginal correlations between observed variables \cite{geng2018joint, leek2007capturing,stegle2011efficient,gagnon2013removing,freytag2015systematic,jacob2016correcting}.  
%In contrast to \citet{chandrasekaran2012latent}, the confounders in \citet{parsana2019addressing} are usually thought of as nuisance variables and would not be included in the complete graph, even if they had been observed.

\newnew{Alternatively, equation \eqref{eqn:splusl_sigma} can be reparametrized as the observed precision matrix by applying the Sherman-Morrison identity} \cite{horn2012matrix} as, 
\begin{align}
\label{eqn:splusl_omega2}
   {\Omegab}_{obs} = {\Sigmab}_{obs}^{-1} = \Omegab - \Lb_{\Omegab},
\end{align} 
where $\Lb_{\Omegab}$ again reflects the effect of unobserved confounding, \newnew{i.e., unobserved nodes in a graph} \cite{chandrasekaran2011rank}.  One such approach, known as latent variable Gaussian Graphical Models (LVGGM), uses parameterization (\ref{eqn:splusl_omega2}) and involves joint inference for $\Omegab$ and $\Lb_{\Omegab}$.  The motivation behind LVGGM is to address the effect of unobserved variables in the complete data graph, which affect the partial correlations of the variables in the observed precision matrix $\Omegab$.  This perspective can be particularly useful when the unobserved variables would have been included in the graph, had they been observed. 

% In order to adjust for the effects of confounding, some key assumptions are required. \new{If $\Omegab$ is indeed sparse, then $\Lb_{\Omegab}$, or equivalently $\Lb_{\Sigmab}$, must be low rank as detailed in \eqref{eqn:compare_thm}.} The low rank assumption is equivalent to assuming that the number of confounding variables is small relative to the number of observed variables.  As such, these methods are often referred to as ``sparse plus low rank'' methods \citep{chandrasekaran2011rank}. \new{LVGGM is essentially a sparse plus low-rank decomposition on the sample precision matrix, and thus, in order for the sparse component $\Omegab$ to be identifiable, the low-rank component $\Lb_{\Omegab}$ must be sufficiently dense \citep{chandrasekaran2011rank}.  }

%In both cases, the observed precision matrix may lead to inference of networks which do not reflect the true underlying relationships among features.  Motivated by this problem, multiple approaches have been proposed to recover the graph encoded by $\Omegab$ in the presence of confounding. 
\new{In previous work, either parameterization  \eqref{eqn:splusl_sigma}, e.g., PCA+GGM,  or \eqref{eqn:splusl_omega2}, e.g., LVGGM, has been used, depending on the source of confounding and motivations as described earlier.} \new{Typically, LVGGM is appropriate when confounding is induced by unobserved nodes in a complete data graph of interest, whereas PCA+GGM is more appropriate when confounding corresponds to nuisance variables, e.g., from batch effects.}

In practice, the selection between these two methods will depend on user's belief about the type of confounding present in the observed data.  In this paper, our goal is to explore a way to address the effect of confounders in order to obtain the graph structure encoded in $\Omegab$ without making such selection. 

To achieve this goal, we generalize two seemingly different methods, PCA+GGM and LVGGM, into a common framework for addressing the effect of $\Lb_{\Sigmab}$ in order to obtain the graph structure encoded in $\Omegab = \Sigmab^{-1}$. Based on the generalization, We propose a new method, PCA+LVGGM, to address two different sources of confounding. The combined approach is more general, since PCA+LVGGM contains both LVGGM and PCA+GGM as special cases. To our knowledge, the two methods of addressing confounding have not been discussed together in the literature. %In this work, we analyze and address the effect of both confounding under one unified framework. 

In summary, in this paper,
\begin{itemize}
  \item we carefully compare PCA+GGM and LVGGM, and illustrate the connection and difference between these two methods. We first theoretically characterize the performance of PCA+GGM. \newnew{Different from \cite{parsana2019addressing} who derives asymptotic results, we provide a non-asymptotic convergence result for the performance of PCA+GGM.} \newnew{We observe that the performance of PCA+GGM are largely determined by the spectral structure of $\Sigmab$ and $\Lsigma$. }%\new{Our theoretical analysis in section \ref{sec:theory} indicates that the spectral norm of $\Sigmab$ and $\Lsigma$ and the inner product between their eigenvectors determine the performance of PCA+GGM and LVGGM. }We demonstrate both theoretically and empirically that PCA+GGM works particularly well \new{when the spectral norm of extraneous low-rank component $\Lsigma$ is large compared to that of the signal $\Sigmab$.}
  
  \item we propose PCA+LVGGM, which combines elements of PCA+GGM and LVGGM.  In simulation, PCA+LVGGM can outperform PCA+GGM or LVGGM when the data is corrupted by multiple confounders. We perform extensive numerical experiments to validate the theory, compare the performance of the three methods, and demonstrate the utility of our approach in two applications.
 
\end{itemize}

The remainder of this paper is organized as follows: In Section \ref{sec:setup}, we introduce the problem definition for GGM, LVGGM and PCA+GGM followed by a brief literature review. Next, we introduce our hybrid method, PCA+LVGGM, and present a novel theoretical results for PCA+GGM in Section \ref{sec:theory}. We use these result to analyze the similarities and differences between LVGGM and PCA+GGM.  In Section \ref{sec:simulation}, we compare the utility of the various approaches in the simulation setting. Finally, in Section \ref{sec:application} we apply the methods on two real world data sets.

We introduce some general notation used throughout the rest of the paper. For a vector $\vb = [v_1, \cdots, v_p]^\tT$, define $\Vert \vb \Vert_2 = \sqrt{\sum_{i=1}^p v_i^2}$ , $\Vert \vb \Vert_1 = \sum_{i=1}^p|v_i|$ and $\Vert\vb \Vert_{\infty} = \max_{i} |v_i|$. For a matrix $\Mb$, let $M_{ij}$ be its $(i,j)$-th entry. Define the Frobenius norm $\Vert \Mb \Vert_F = \sqrt{\sum_i\sum_j M_{ij}^2}$, the element-wise $\ell_1$-norm  $\Vert \Mb \Vert_1 = \sum_{i}\sum_{j}|M_{ij}|$ and  $\Vert \Mb \Vert_{\infty} = \max_{(i,j)}|M_{ij}|$. We also define the spectral norm $\Vert \Mb \Vert_2 = \sup_{\Vert \vb \Vert_2 \le 1}\Vert \Mb \vb\Vert_2$ and $\Vert\Mb \Vert_{L_1} = \max_j \sum_i|M_{ij}|$. The nuclear norm $\Vert \Mb \Vert_*$ is defined as the sum of the singular values of $\boldsymbol{M}$. When $\Mb \in \mathbb{R}^{p \times p}$ is symmetric, its eigendecomposition is $\Mb=\sum_{i=1}^p \lambda_i \vb_i \vb_i ^\tT$, where $\lambda_i$ is the $i$-th eigenvalue of $\Mb$, and $\vb_i$ is the $i$-th eigenvector. We assume that $\lambda_1 \ge \cdots \ge \lambda_p$. We call $\lambda_i \vb_i \vb_i ^\tT$ the $i$-th eigencomponent of $\Mb$.

\section{Problem setup and review} \label{sec:setup}
\subsection{Gaussian graphical models} \label{sec:ggm}
Consider a $p$-dimensional random vector $\Xb=[X_1,\cdots,X_p]^\text{T}$ with covariance matrix $\Sigmab$ and precision matrix $\Omegab$. Let $G=(V,E)$ be the graph associated with $\Xb$, where $V$ is the set of nodes (or vertices) corresponding to the elements of $\Xb$, and $E$ is the set of edges connecting nodes. The graph shows the conditional independence relations between elements of $\Xb$. For any pair of connected nodes, the corresponding pairs of variables in $\Xb$ are conditionally independent given the rest variables, i.e., $X_i \indep X_j | X_{ \backslash i,j}$, for all $(i,j) \notin E$. If $\Xb$ is multivariate Gaussian, then $X_i$ and $X_j$ are conditionally independent given other variables if and only if $\Omegab_{ij}=0$, and thus the graph structure can be recovered from the precision matrix of $\Xb$.

Without loss of generality, we assume the variable $\Xb$ has mean zero in this paper. Assuming that the graph is sparse, given a random sample $\{\Xb^{(1)},\cdots,\Xb^{(n)}\}$ following the distribution of $\Xb$, the Glasso estimate $\hat{\Omegab}_{Glasso}~$\cite{yuan2007model,friedman2008sparse} is obtained by solving the following log-likelihood based $\ell_1$-regularized function:
\begin{mini}
{\Omega\succ 0 }{\text{Tr}(\Omegab \Sigmab_n) - \lnn{\text{det} (\Omegab)}+\lambda \Vert \Omegab \Vert_1}
{\label{glasso-objective}}{},
\end{mini}
where $\text{Tr}$ denotes the trace of a matrix and $\Sigmab_n=({1}/{n})\sum_{k=1}^{n} \Xb^{(k)} {\Xb^{(k)}}^\text{T}$ is the sample covariance matrix. Many alternative objective functions for sparse precision matrix estimation have been proposed \cite{cai2011constrained, meinshausen2006high, peng2009partial, khare2015convex}.  The behavior and convergence rates of these approaches are well studied \cite{bickel2008regularized,bickel2008covariance,rothman2008sparse,lam2009sparsistency,ravikumar2011high,cai2016estimating}. 

% propose an estimator for sparse estimation of precision matrices called CLIME, which achieves the same order of convergence rate as Glasso. Another family of approaches related to edge selection is neighborhood selection, \citet{meinshausen2006high,peng2009partial,khare2015convex} provides more detail.

% The performance of those methods are well studied.  \citet{bickel2008regularized,bickel2008covariance,rothman2008sparse} study the convergence rate of the $l_1$-regularized maximum likelihood estimation of precision matrix, \citet{ravikumar2011high} obtains the convergence rate in the elementwise $l_{\infty}$ norm, which can infer the edge selection (also called 'sparsistency'). Efficient algorithms were developed for solving the $l_1$-regularized problem \citep{banerjee2008model,friedman2008sparse,hsieh2014quic}. 

% \subsection{{Gaussian Graphical Models with Latent Variables}} \label{sec22}

In presence of latent confounders, Glasso and other GGM methods would likely recover a more dense precision matrix owing to spurious partial correlations introduced between observed variables. In other words, even when the underlying graph is sparse conditioned on the latent variables, the observed graph is dense marginally.

\subsection{Latent variable Gaussian graphical models} \label{sec:lvggm}

One method for controlling the effects of confounders is the Latent Variable Gaussian Graphical Model (LVGGM) approach first proposed  by \cite{chandrasekaran2012latent}.  They assume that the number of latent factors is small compared to the number of observed variables, and that the conditional dependencies among the observed variables conditional on the latent factors is sparse.  Consider a $(p+r)$ dimensional mean-zero normal random variable $\Xb=[\Xb_O, \Xb_H]^\tT$, where $\Xb_O \in \mathbb{R}^p$ is observed and $\Xb_H \in \mathbb{R}^r$ is latent. Let $\Xb$ have precision matrix $\Omegab \in \mathbb{R}^{(p+r)\times(p+r)}$ , and the submatrices $\Omegab_O \in \mathbb{R}^{p \times p}$, $\Omegab_H\in \mathbb{R}^{r \times r}$ and $\Omegab_{O,H}\in \mathbb{R}^{p \times r}$ specify the dependencies between observed variables, between latent variables and between the observed and latent variables respectively. By Schur complement, the inverse of the observed covariance matrix satisfies:
\begin{equation}
\label{eqn:omega_obs}
    {\Omegab}_{obs}={\Sigmab}^{-1}_{obs}=\Omegab_O - \Omegab_{O,H}\Omegab_H^{-1}\Omegab_{O,H}^\text{T}=\Omegab-\Lb_{\Omegab}.
\end{equation}
where $\Omegab=\Omegab_O$ encodes the conditional independence relations of interest and is sparse by assumption. $\Lb_{\Omegab}=\Omegab_{O,H}\Omegab_H^{-1}\Omegab_{O,H}^\tT$ reflects the low-rank effect of latent variables $\Xb_H$.  Based on this sparse plus low-rank decomposition,  \cite{chandrasekaran2012latent} proposed the following problem:
\begin{mini}[2]
{\Omegab, \Lb_{\Omegab}}{-\ell(\Omegab-\Lb_{\Omegab};{\Sigmab}_n)+\lambda \Vert \Omegab \Vert_1 +\gamma \Vert \Lb_{\Omegab} \Vert_*}
{\label{lvggm-objective}}{}
\addConstraint{\Lb_{\Omegab}}{\succeq \boldsymbol{0},~~~~}{\Omegab-\Lb_{\Omegab}}~{\succ \boldsymbol{0},}
%\addConstraint{\Omegab-\Lb_{\Omegab}}{\succ 0},
\end{mini}
where $\Sigmab_n$ is the observed sample covariance matrix and $\ell(\Omegab,\Sigmab)=\lnn{(\det{(\Omegab))}}-\text{Tr}(\Omegab \Sigmab)$ is the Gaussian log-likelihood function. The $\ell_1$-norm encourages sparsity on $\Omegab$ and the nuclear norm encourages low-rank structure on $\Lb_{\Omegab}$. 

The sparse plus low-rank decomposition is ill-posed if $\Lb_{\Omegab}$ is not dense.  If \new{$\Lomega$ is sparse, then it} is indistinguishable from $\Omegab$, \new{that is, the sparse plus low-rank decomposition works well only when the sparse component is not low-rank and the low-rank component is not sparse \cite{chandrasekaran2011rank}. In practice, $\Lomega$ is dense if the latent variables have widespread effects.}. Identifiability of $\Omegab$ coincides with the incoherence condition in the matrix completion problem \cite{candes2009exact} which requires that $\vert \vb_k ^\tT \eb_i \vert$ is small
% \le \alpha, \ \ \ for \ all \ pairs \ of \ k=1,..r \ and \ i=1,...p,
for all $k \in \{1,\cdots, r\}$ and $i \in \{1,\cdots, p\}$ where $\vb_k$ is the $k$-th eigenvector of $\Lb_{\Omegab}$ and $\eb_i$ is the $i$-th standard basis vector.  %When the eigenvectors of $\Lb_{\Omegab}$ are sufficiently dissimilar to the canonical vectors, the error bound\footnote{Actually, in \citet{chandrasekaran2012latent}, the bound is for $\max\{\Vert \Omegab - \hat{\Omegab} \Vert_{\infty}, \Vert \Lb_{\Omegab} - \hat{\Lb}_{\Omegab} \Vert_{2}\}$. Since $\Omegab$ and $\Lb_{\Omegab}$ are estimated jointly, we have to use that bound even if we want to bound $\Vert \Omegab - \hat{\Omegab} \Vert_{\infty}$ individually.}, $\Vert \Omegab - \hat{\Omegab} \Vert_{\infty}$, is of order $O(\sqrt{\frac{p}{n}})$. 
More analysis on LVGGM can be found in \cite{agarwal2012noisy, meng2014learning}.

Finally, \cite{ren2012discussion} shows that the standard GGM approaches can still recover $\Omegab$ in the presence of latent confounding as long as the spectral norm of the low-rank component is sufficiently small compared to that of $\Sigmab$. This is also verified in our simulations. %In our simulations, we observe that LVGGM  significantly outperforms Glasso only \new{when the low-rank component $\Lsigma$ has large enough spectral norm compared to the signal $\Sigmab$}. If the magnitude of $\Lb_{\Omegab}$ is large enough, however,  a more direct approach that remove high variance components can be more appropriate. This partly motivates the PCA+GGM method described in the next section.
%  when the norm of $L_\Omega$  
% is not small, the Glasso procedure can also recover $\Omega_O$. In other words, those $\ell_1$-regularized methods are robust when the noise is small. Therefore,

\subsection{PCA+GGM} \label{sec:pca+ggm}
Unlike LVGGM, which involves a decomposition of the observed data precision matrix, PCA+GGM involves a decomposition of the observed data covariance matrix:
\begin{equation}
\label{eqn:sigma_obs}
{\Sigmab}_{obs}={\Omegab}_{obs}^{-1}=(\Omegab-\Lb_{\Omegab})^{-1}=\Omegab^{-1}+\Lb_{\Sigmab}.
\end{equation}
Motivated by  confounding from measurement error and batch effects, \cite{parsana2019addressing} proposed the principal components correction (PC-correction) for removing $\Lb_{\Sigmab}$. Consider observed data $\Xb_{obs}$, such that
\begin{equation}
\label{eqn:xplusl}
    {\Xb}_{obs}=\Xb+\Ab\Zb,
\end{equation}
where $\Xb \sim N(\boldsymbol{0}, \Sigmab)$ and $\Zb \sim N(\boldsymbol{0}, \Ib_r)$. Matrix $\Ab \in \mathbb{R}^{p \times r}$ is non-random so that $\Lb_{\Sigmab}=\Ab\Ab^\tT$. In general, additional structural assumptions are needed to distinguish $\Lb_{\Sigmab}$ from $\Sigmab$. \newnew{As we will discuss in Section \ref{sec:theory},} \newnew{one of our contributions is to show that }if the \new{spectral} norm of $\Lb_{\Sigmab}$ is large relative to that of $\Sigmab$, then under mild conditions, $\Lb_{\Sigmab}$ is close to the sum of the first few eigencomponents of ${\Sigmab}_{obs}$. Therefore,  one can remove the first $r$ eigencomponents from ${\Sigmab}_{obs}$ \cite{parsana2019addressing}. This PCA+GGM method is described in Procedure \ref{procedure:PCAGGM}. Note that the number of principal components needs to be determined a priori, which we discuss in subsequent sections. 

\begin{figure}
\begin{minipage}[t]{0.49\textwidth}
\begin{algorithm}[H]
  \caption{PCA+GGM}
  \begin{algorithmic}[1]
    \Statex {\bf Input}: Sample covariance matrix, $\hat\Sigmab_{obs} = ({1}/{n})\sum_{k=1}^{n} \Xb_{obs}^{(k)} {\Xb_{obs}^{(k)}}^\tT$; rank of $\hat \Lb_{\Sigmab}$, $r$ 
    \Statex {\bf Output}: Precision matrix estimate, $\hat\Omegab$
    \State Estimate $\hat \Lb_{\Sigmab}$ from eigencomponents:
        $$\hat{{\Sigmab}}_{obs} = \sum_{i=1}^{p}{\hat{\lambda}_i \hat{\thetab}_i \hat{\thetab}_i^\tT},\ \hat \Lb_{\Sigmab} = \sum_{i=1}^{r}{\hat{\lambda}_i \hat{\thetab}_i \hat{\thetab}_i^\tT}$$
    \State Remove $\hat \Lb_{\Sigmab}$:
        $$\hat{\Sigmab} = \hat{{\Sigmab}}_{obs} - \hat \Lb_{\Sigmab}.$$
    \State Using $\hat\Sigmab$, compute $\hat\Omegab$ as solution to \eqref{glasso-objective} 
  \end{algorithmic}
  \label{procedure:PCAGGM}
\end{algorithm}
\end{minipage}
\hfill
\begin{minipage}[t]{0.49\textwidth}
\begin{algorithm}[H]
  \caption{PCA+LVGGM}
  \begin{algorithmic}[1]
    \Statex {\bf Input}: Sample covariance matrix, $\hat\Sigmab_{obs}$; rank of $\hat \Lb_{\Sigmab}$, $r_P$; rank of $\hat \Lb_{\Omegab}$, $r_L$ 
    \Statex {\bf Output}: Precision matrix estimate, $\hat\Omegab$
    \State Estimate $\hat \Lb_{\Sigmab}$ from eigencomponents:
        $$\hat{{\Sigmab}}_{obs} = \sum_{i=1}^{p}{\hat{\lambda}_i \hat{\thetab}_i \hat{\thetab}_i^\tT},\ \hat \Lb_{\Sigmab} = \sum_{i=1}^{r_P}{\hat{\lambda}_i \hat{\thetab}_i \hat{\thetab}_i^\tT}$$
    \State Remove $\hat \Lb_{\Sigmab}$:
        $$\hat{\Sigmab} = \hat{{\Sigmab}}_{obs} - \hat \Lb_{\Sigmab}.$$
    \State Using $\hat\Sigmab$, compute $\hat\Omegab$ as solution to \eqref{lvggm-objective} with $\gamma$ such that $\text{rank}(\hat \Lb_{\Omegab}) = r_L$
  \end{algorithmic}
  \label{procedure:PCALVGGM}
  \end{algorithm}
\end{minipage}
%\caption{Procedures of PCA+GGM and PCA+LVGGM.}
%\label{fig-procedure}
\end{figure}

%% $\mathbf{Procedure \ 1}$ - PCA+GGM: 
%% \begin{enumerate}[Step 1.]
%% \label{pcd1}
%% \item \textcolor{red}{Given the number of principal components to be removed $r$, remove the first $r$ eigencomponents and obtain $\hat{\Sigma}$:}
%% $$\hat{{\Sigma}}_{obs} = \sum_{i=1}^{p}{\hat{\lambda}_i \hat{\theta}_i \hat{\theta}_i^T}, \ \ \ \   \hat{\Sigma} = \hat{{\Sigma}}_{obs} - \sum_{i=1}^{r}{\hat{\lambda}_i \hat{\theta}_i \hat{\theta}_i^T}.$$
%% \item Run Glasso on $\hat{\Sigma}$ to obtain the estimate of $\Omega$.
%% \end{enumerate}

% We demonstrate the effectiveness of Procedure 1 both theoretically in (Section \ref{sec:theory}) and in simulations (Section \ref{sec:simulation}). 

\subsection{{Combining PCA+GGM and LVGGM}} \label{sec:pca+lvggm}

%Although both LVGGM and the PCA+GGM are methods for solving the same problem, they have different motivations. As described in previous sections, identifiability in LVGGM models requires that $\Lb_{\Omegab}$ is sufficiently dense, whereas identifiability in PCA+GGM requires that the norm of $\Lb_{\Sigmab}$ dominates the norm of $\Sigmab$.  \new{These assumptions are satisfied in many applications. For example, in financial applications, $\Lomega$ will be dense if the the latent variables have a widespread effect across many stocks \cite{chandrasekaran2012latent, fama2004capital}.  The large norm assumptions can be satisfied if the latent confounding effect is strong, i.e., the eigenvalues structure is very spiky and the top eigenvalues reflect the effect of confounding \cite{johnstone2001distribution}. }
\newnew{As previously mentioned, while LVGGM and PCA+GGM solve the same problem, they are motivated by different sources of confounding.} In applications, the observed data may be corrupted by  multiple  sources of confounding, and thus elements from both methods are needed. For example, in the biological application discussed in Section \ref{sec:gene}, both batch effects and unmeasured biological variables likely confound estimates of graph structure between observed variables. \newnew{This motivates us to propose the PCA+LVGGM strategy described below.} 

As (\ref{eqn:omega_obs}) illustrated, the observed precision matrix $\Omegab^{'}$ may have been corrupted by a latent factor $\Lb_{\Omegab}$:
\begin{equation}
    \label{eqn:spluslomega_prime}
    {\Omegab}^{'}=\Omegab-\Lb_{\Omegab}.
\end{equation}
Now, rewriting (\ref{eqn:spluslomega_prime}) in terms of $\Sigmab= \Omegab^{-1}$ and ${{{\Sigmab}}^{'}}= {{{\Omegab}}^{'}}^{-1}$, applying the Sherman-Morrison identity on $\Omegab^{'}$ gives,
\begin{equation}
    \label{eqn:covpluslomega_prime}
    {\Sigmab}^{'}=\Sigmab + \Lb_{\Omegab}^{'},
\end{equation}
where $\Lb_{\Omegab}^{'}$ is still a low-rank matrix. If ${\Sigmab}^{'}$ is further corrupted by an additive latent factor represented by $\Lb_{\Sigmab}$, the following equation described the observed matrix ${\Sigmab}_{obs}$:
\begin{equation}
    \label{eqn:sum-of-low-rank-components}
    {{\Sigmab}_{obs}}={\Sigmab}^{'}+\Lb_{\Sigmab}=\Sigmab + \Lb_{\Omegab}^{'} +\Lb_{\Sigmab}
\end{equation}
In the above example, \newnew{following our theoretical analysis in Section \ref{sec:theory}}, if the spectral norm of $\Lb_{\Sigmab}$ is much larger than that of $\Sigmab$ and $\Lb_{\Omegab}^{'}$, then removing $\Lb_{\Sigmab}$ using the PC-correction is likely to be effective. \new{If the spectral norm of $\Lb_{\Omegab}^{'}$ is not much larger than that of $\Sigmab$, then PC-correction is not a good choice to remove $\Lb_{\Omegab}^{'}$. If $\Lb_{\Omegab}^{'}$ is dense, then $\Omegab$ and $\Lb_{\Omegab}$ can be well estimated by LVGGM.} In \eqref{eqn:sum-of-low-rank-components}, the overall confounding $\Lb_{\Omegab}^{'}+\Lb_{\Sigmab}$ is the sum of two low-rank components with different norms, we can consider using both methods: first remove $\Lb_{\Sigmab}$ via eigendecomposition, then apply LVGGM to estimate $\Omegab$ and $\Lb_{\Omegab}$. We call this procedure PCA+LVGGM and it is shown in Procedure \ref{procedure:PCALVGGM}. We discuss methods for setting \new{the ranks for $\Lb_{\Sigmab}$ (defined as $r_P$) and $\Lb_{\Omegab}^{'}$ (defined as $r_L$)} in Section \ref{sec:tuning}. 

\section{Theoretical analysis and model comparisons} \label{sec:theory}
In this section, we investigate the theoretical properties of PCA+GGM. \new{Our results reveal precisely how the eigenstructure of the observed covariance matrix affects the the performance of PCA+GGM. The theoretical analysis provides practical insights into when each graph estimation method should (or should not) be applied.} Specifically, we derive the convergence rate of PCA+GGM and compare it to that of LVGGM. As shown in theoretical analysis by \cite{parsana2019addressing}, the low-rank confounder can be well estimated by PC-correction if the number of features $p \to \infty$ with the number of observations $n$ fixed. We provide a non-asymptotic analysis depending on $p$ and $n$ and our result shows that the graph can be recovered exactly when $n \to \infty$ with fixed $p$. \new{When additional assumptions are satisfied, e.g. spiky covariance structure and incoherent eigenvectors, the convergence rate can be improved to $O(\sqrt{{\lnn p}/{n}})$}.

\subsection{Convergence analysis on PCA+GGM} \label{sec:convergence}
Without loss of generality, we consider the case of a rank-one confounder. Assume that we have a random sample of $p$-dimensional random vectors:
\begin{equation}
    \label{eqn:xplusl_proof}
    {\Xb}^{(i)}_{obs}=\Xb^{(i)}+{\sigma}\vb \Zb^{(i)}, \ i \in \{1,\cdots,n\},
\end{equation}
where $Cov(\Xb^{(i)})=\Sigmab$ and $\Zb^{(i)}$ is a univariate standard normal random variable. $\vb \in \mathbb{R}^p$ is a non-random vector with unit norm, and $\sigma$ is a non-negative scalar constant. Without loss of generality, we assume that 
$\Xb^{(i)} \indep \Zb^{(i)}$. To see how $\vb$ affects estimation, we assume that $\vb$ is the $k$-th eigenvector of $\Sigmab$. \newnew{The discussion on general $\vb$ is deferred to Section \ref{sec:analysis_thm}.} Therefore, the covariance matrix of ${\Xb}^{(i)}_{obs}$ is:
\begin{equation}
    \label{eqn:sigmaobs_proof}
    {\Sigmab}_{obs}=\Sigmab + \sigma^2 \vb \vb^\text{T} = \Sigmab_{-k}+(\lambda_k(\Sigmab)+\sigma^2)\vb \vb^\tT,
\end{equation}
where $\Sigmab_{-k}$ is the matrix $\Sigmab$ without the $k$-th eigencomponent and $\lambda_k(\Sigmab)$ is the $k$-th eigenvalue of $\Sigmab$. When $\sigma^2 > \lambda_1(\Sigmab)$, $\lambda_k(\Sigmab)+\sigma^2$ becomes the first eigenvalue of ${\Sigmab}_{obs}$, and $\vb$ is the corresponding first eigenvector. We remove the first principal component from the sample covariance matrix $\hat{{\Sigmab}}_{obs}$:
\begin{equation}
    \label{eqn:sigmaest_proof}
    \hat{\Sigmab}=\hat{{\Sigmab}}_{obs}-\hat{\lambda}_1 \hat{\thetab}_1 \hat{\thetab}_1^\tT,
\end{equation}
where $\hat{\lambda}_1$ is the first eigenvalue of $\hat{{\Sigmab}}_{obs}$ and $\hat{\thetab}_1$ is the first eigenvector of $\hat{{\Sigmab}}_{obs}$. Then we use $\hat{\Sigmab}$ to estimate $\Omegab$. We first show that under mild conditions, $\hat{\Sigmab}$ is close to $\Sigmab$. Following \cite[3.1]{bickel2008regularized}, \newup{we assume that $\Sigmaobs$ is well-conditioned such that:}
\begin{equation}
    \label{eqn:sigmaeig_proof}
    \newup{0 < \epsilon \le \lambda_p(\Sigmab_{obs}) \le \lambda_1(\Sigmab_{obs}) \le \frac{1}{\epsilon},}
\end{equation}
\newup{where $\epsilon$ is independent of $p$.}
\begin{theorem}
\label{thm1}
Let $\lambda_i$ be the $i$-th eigenvalue of $\Sigmab_{obs}$, \newup{$\thetab_i$ be the  $i$-the eigenvector of $\Sigmab_{obs}$,} and $\nu = \lambda_1 - \lambda_2$ be the eigengap of ${\Sigmab}_{obs}$. Suppose $\Sigmab_{obs}$ satisfies condition (\ref{eqn:sigmaeig_proof}) and ${\Xb}^{(i)}_{obs}$ is generated as (\ref{eqn:xplusl_proof}). Further assume that $\sigma^2 > \lambda_1(\Sigmab)$. Suppose $n \ge p$ and $\Vert \Sigmab \Vert_2 \sqrt{{(\nu + 1)}/{\nu^2}}\sqrt{{p}/{n}}\le {1}/{128}$, then:
\begin{align*}
    \Vert \hat{\Sigmab} - \Sigmab \Vert_{\infty} \le C_1\sqrt{\frac{\lnn p }{n}} + C_2\sqrt{\frac{\nu+1}{\nu^2}}\sqrt{\frac{p}{n}} + C_3 \sqrt{\frac{p}{n}} + \lambda_k(\Sigmab)\Vert \thetab_1 \thetab_1^\tT \Vert_{\infty},
\end{align*}
with probability greater than $1-C_4/p$ for constants $\Cis$.
\end{theorem}
\begin{proof}[\textbf{\upshape Proof:}] 
\newup{Recall when $\sigma^2 > \lambda_1(\Sigmab)$, $\lambda_k(\Sigmab)+\sigma^2$ becomes the first eigenvalue of ${\Sigmab}_{obs}$, and $\vb$ becomes its first eigenvector.} By \eqref{eqn:sigmaobs_proof} and \eqref{eqn:sigmaest_proof},
\begin{align*}
    \hat{\Sigmab}-\Sigmab &= (\hat{{\Sigmab}}_{obs}-\hat{\lambda}_1 \hat{\thetab}_1 \hat{\thetab}_1^\tT) - ({\Sigmab}_{obs} - \lambda_1 \thetab_1 \thetab_1^\tT +  \newup{\lambda_k(\Sigmab) \thetab_1 \thetab_1^\tT)}\\
    &= (\hat{{\Sigmab}}_{obs}-{\Sigmab}_{obs})+(\lambda_1 \thetab_1 \thetab_1^\tT - \hat{\lambda}_1 \hat{\thetab}_1 \hat{\thetab}_1^\tT) - \newup{\lambda_k(\Sigmab) \thetab_1 \thetab_1^\tT.}
\end{align*}
At a high level, we bound $\Vert \hat{\Sigmab}-\Sigmab \Vert_{\infty}$ by bounding the norms of ${\Sigmab}_{obs}-\hat{{\Sigmab}}_{obs}$, $\lambda_1-\hat{\lambda}_1$ and $\thetab_1-\hat{\thetab}_1$. The details of the complete proof is in Appendix \ref{pf-outline}.
\end{proof}

The bound in Theorem \ref{thm1} can be further simplified as $C_s\sqrt{{p}/{n}}+\lambda_k(\Sigmab)\Vert \thetab_1 \thetab_1^\tT \Vert_{\infty}$ for some large constant $C_s$. We express it in the above form because it provides more insight on how each term affects the result. Now we analyze the bound in Theorem \ref{thm1} in detail.

The  error bound in Theorem \ref{thm1} depends on the largest eigenvalue of ${\Sigmab}_{obs}$, the eigengap $\nu = \lambda_1({\Sigmab}_{obs})-\lambda_2({\Sigmab}_{obs})$, the eigenvector of the confounder and $n$ and $p$. %\new{Additionally, we assume that the norm of confounding should be large enough compared to that of $\Sigmab$, i.e., $\sigma^2$ in \eqref{eqn:sigmaobs_proof} is bigger than $\lambda_1(\Sigmab)$. The reason for this assumption is that if $\sigma^2$ is not large enough compared to $\lambda_1(\Sigmab)$, then the first eigencomponent of ${\Sigmab}$ will be removed, leading to a poor estimation.} 
The term $\sqrt{{(\nu + 1)}/{\nu^2}}$ shows that if the eigengap $\nu$ is larger, the estimation error bound will be smaller. Recall that when $\sigma^2 > \lambda_1(\Sigmab)$, $\lambda_k(\Sigmab)+\sigma^2$ becomes the first eigenvalue of $\Sigmab_{obs}$. Hence if $\sigma^2 \gg \lambda_1(\Sigmab)$, then the eigengap $\nu$ is large. The fact that a larger eigengap leads to a better convergence rate is closely related to the concept of ``effective dimension'' (also known as ``effective rank''). The effective rank, $r(\boldsymbol{M})$, of any positive semidefinite matrix $\Mb \in \mathbb{R}^{p \times p}$,  \new{is defined as}:
\begin{equation}
    \label{eqn:eff_rank}
    r(\boldsymbol{M}) := \frac{\text{Tr}(\Mb)}{\lambda_1(\Mb)}=\frac{\sum_{i=1}^{p}\lambda_i(\Mb)}{\lambda_1(\Mb)}\le C,
\end{equation}
where $C \ge 1$ can be viewed as the effective dimension of $\Mb$ \cite{meng2014learning,koltchinskii2017concentration,wainwright2019high}. $M$ is approximately low-rank if the first few eigenvalues are much larger than the rest, and $r(\boldsymbol{M})$ will be much smaller than the observed dimension $p$.  In this case, we can significantly reduce the magnitude of the dependence on  $O(\sqrt{{p}/{n}})$ by replacing $p$ with effective dimension $C$, in (\ref{eqn:eff_rank}). \new{We provide a sharper bound for matrices with small effective rank in Theorem \ref{thm1new}.}

Next, we reason about the last term in the error bound, $\lambda_k(\Sigmab)\Vert \thetab_1 \thetab_1^\tT \Vert_{\infty}$. \new{In practice, in many sparse graphs inferred from real world data, the first few eigenvalues of $\Sigmab$ are much larger than the rest, i.e., $\lambda_1(\Sigmab) \gg \lambda_k(\Sigmab)$ for large enough $k > 1$. This is also true for many common graph data generating models (see Appendix \ref{app-eigen}). }This means that if the eigenvector of the low-rank component is one of the first few eigenvectors of $\Sigmab$, then the error bound will be much larger. This result shows that the first few eigencomponents play a more important role in determining the structure of $\Sigmab$ and its inverse. Thus, the error of the PCA+GGM estimator will be large if those first few eigencomponents are removed by PC-correction. 

Note that $\Vert \thetab_1 \thetab_1^\tT \Vert_{\infty}$ is upper bounded by 1, since $\thetab_1$ is the eigenvector of some matrix, and thus has unit Euclidean norm; however, $\Vert \thetab_1 \thetab_1^\tT \Vert_{\infty}$ can be much smaller than 1 when $\thetab_1$ is \new{incoherent with standard basis, \newnew{e.g. dense}}. One extreme case is when all the elements of $\thetab_1$ are ${1}/{\sqrt{p}}$, in which case $\Vert \thetab_1 \thetab_1^\tT \Vert_{\infty} = {1}/{p}$. This setup corresponds to a scenario in which the confounder has a widespread effect over all the $p$ variables in the signal, which is in accordance with one requirement in LVGGM. LVGGM requires the low-rank component to be dense. For both PCA+GGM and LVGGM, more "widespread" confounding implies smaller estimation error. \new{Based on these observations, we provide a tighter bound under small effective rank and incoherent $\thetab_1$.}

\new{
\begin{theorem}
\label{thm1new}
Following the same notations and assumptions for $\Sigmab_{obs}$ in Theorem \ref{thm1} and again assume that $\sigma^2 > \lambda_1(\Sigmab)$. Further assume that \newup{the effective rank of $\Sigmab_{obs}$ (defined in \eqref{eqn:eff_rank}) $r(\Sigmab_{obs}) \asymp p/\lambda_1$, the eigengap $\nu$ satisfies $\lambda_1 \asymp \nu \gg p\sqrt{\lnn p/n}$}, and $\thetab_1$ is incoherent, i.e. $\| \thetab \|_\infty \le C_1/\sqrt{p}$. Then:
\begin{align*}
    \newup{\Vert \hat{\Sigmab} - \Sigmab \Vert_{\infty} \le C_2(\sqrt{\frac{\lnn p }{n}} + \frac{1}{p})},
\end{align*}
with probability greater than $1-C_3/p$ for some $\Cis$. 
\end{theorem}
\begin{proof}[\textbf{\upshape Proof:}] 
The complete proof is in Appendix \ref{pf-outline}.
\end{proof}
}

After obtaining $\hat{\Sigmab}$, we can use Glasso, CLIME \cite{cai2011constrained} or any sparse GGM estimation approach to estimate $\Omegab$. We can have a good estimate of $\Omegab$ when $\Vert \hat{\Sigmab} - \Sigmab \Vert_{\infty}$ is small. With the same input $\hat{\Sigmab}$, the theoretical convergence rate of the estimate obtained from CLIME is of the same order as the Glasso estimate. The derivation of the error bound of Glasso requires the irrepresentability condition and restricted eigenvalue conditions (see \cite{ravikumar2011high}). Due to the length of the article, we only show the proof of the edge selection consistency for CLIME, meaning that for the theoretical analysis, we apply CLIME method after obtaining $\hat{\Sigmab}$. 

The CLIME estimator $\hat{\Omegab}_1$ is obtained by solving:
\begin{mini}
{\Omegab}{\Vert \Omegab \Vert_1,~~~~~\text{subject to~~~~}\Vert \hat{\Sigmab} \Omegab - \Ib \Vert_{\infty} \le \lambda_n.}
{\label{eqn:obj_clime}}{}
%\addConstraint{\Vert \hat{\Sigmab} \Omegab - \Ib \Vert_{\infty}}{\le \lambda_n}.
\end{mini}
Since $\hat{\Omegab}_1$ might not be symmetric, we need the symmetrization step to obtain $\hat{\Omegab}$.

Following \cite{cai2011constrained}, we assume that $\Omegab$ is in the following class:
\begin{equation}
    \label{eqn:class_clime}
    U(s_0,M_0)=\{\Omegab=\omega_{ij}: \ \Omegab \succ \boldsymbol{0}, \Vert \Omegab \Vert_{L_1}<M_0, \max_{1 \le i \le p} \sum_{i=1}^p I_{\{\omega_{ij} \ne 0\}}\le s_0(p)\},
\end{equation}
where we allow $s_0$ and $M_0$ to grow as $p$ and $n$ increase. With $\hat{\Sigmab}$ obtained from equation (\ref{eqn:sigmaest_proof}) as the input of (\ref{eqn:obj_clime}), we have the following result.

\begin{theorem}
\label{thm2}
Suppose that assumptions in Theorem \ref{thm1} hold, $\Omegab \in U(s_0, M_0)$, and $\lambda_n$ is chosen as
\begin{align*}
    M_0(C_1\sqrt{\frac{\lnn p }{n}}+ C_2 \sqrt{\frac{p}{n}}+C_3 \sqrt{\frac{\nu+1}{\nu^2}}\sqrt{\frac{p}{n}}+\lambda_k(\Sigmab)\Vert \thetab_1 \thetab_1^T \Vert_{\infty}),
\end{align*}
then:
\begin{align*}
    \Vert \Omegab-\hat{\Omegab} \Vert_{\infty} \le 2 M_0 \lambda_n,
\end{align*}
with probability greater than $1-C_4/p$. $C_i$'s are defined the same as in Theorem \ref{thm1}.

\new{When assumptions in Theorem \ref{thm1new} hold, $\Omegab \in U(s_0, M_0)$, and $\lambda_n^{'}$ is chosen as $M_0(C_5\sqrt{{\lnn p}/{n}})$, then:
\begin{align*}
    \Vert \Omegab-\hat{\Omegab} \Vert_{\infty} \le 2 M_1 \lambda_n^{'},
\end{align*}
with probability greater than $1-C_6/p$ for $\Cis$.}
\end{theorem}
\begin{proof}[\textbf{\upshape Proof:}] 
The main steps follow the proof of Theorem 6 in \cite{cai2011constrained}. The complete proof is in Appendix \ref{pf-thm2}.
\end{proof}

Therefore, if the minimum magnitude of $\Omegab$ is larger than the error bounds above, then we can have exact edge selection with high probability.

\subsection{\new{Generalizations}} \label{sec:analysis_thm}
The analysis in previous sections assumes that the low-rank confounder has rank 1, is independent of $\Xb$ and the eigenvector of the covariance of the $\Lsigma$ is one of the eigenvectors of $\Sigmab$. We now comment on more general settings.
\begin{itemize}
    \item \new{\textit{Higher rank:} For ease of interpretation, we assume that the confounder can be expressed as $\sum_{i=1}^r\sigma_i\vb_i$.  If $\min_{i}\sigma_i^2 > \lambda_{1}(\Sigmab)$, then when running PCA+GGM, the low-rank component can be removed due to its larger norm compared to that of $\Sigmab$.} %we can avoid removing useful signal from $\Sigmab$.} %From the error bound in Theorem \ref{thm1}, we can see that higher ranks will lead to higher upper bound, since more useful signal could be potentially removed. 
    \new{According to Theorem \ref{thm1}, PCA+GGM can still perform well if $\vb_i$'s are not the top eigenvectors of $\Sigmab$.}
    
    \item \new{\textit{\newup{Arbitrary $\vb$:}} \newup{In Theorem \ref{thm1}, $\vb$ is assumed to be an eigenvector of $\Sigmab$, but the result in Theorem \ref{thm1} provides insights about the bound with an arbitrary $\vb$. Assume that $\vb_1$ is the eigenvector of confounder 1 and $\vb_2$ is the eigenvector of confounder 2. $\vb_1$ is the $i$-th eigenvector of $\Sigmab$ and $\vb_2$ is the $j$-th eigenvector of $\Sigmab$. In Theorem \ref{thm1}, we show that the bound for estimating $\Sigmab$ with confounder 1 is larger than thr bound with confounder 2 if $i < j$ because $\vb_1$ is associated with a larger eigenvalue. If a confounder's eigenvector $\vb$ is a linear combination of $\vb_1$ and $\vb_2$, the bound should be between the bound with confounder 1 and the bound with confounder 2. Following this idea, we now show how our analysis works for an arbitrary $\vb$ in more detail.} When the eigenvector of the low-rank component is not one of the eigenvectors of $\Sigmab$, we can express that vector using the eigenvectors of $\Sigmab$ as basis. For example, we assume that in \eqref{eqn:xplusl_proof}, $\vb = \sum_{i=1}^p a_i\thetab_i$, where $\thetab_i$ means the $i$-th eigenvector of $\Sigmab$. We say $\vb$ is closely aligned with $\thetab_1$ if $|a_1|$ is significantly large compared with other $|a_i|$'s. Equivalently, $|\vb^\tT\thetab_1| \gg |\vb^\tT\thetab_i|$ for $i \ne 1$, if $\vb$ is closely aligned with $\thetab_1$. In this case, the first eigencomponent of $\Sigmab$ will be removed, thus leading to a poor estimate of $\Sigmab$ using PC-correction. If the eigenvector of the low-rank component is not closely aligned with the first few eigenvectors of $\Sigmab$, then we won't lose too much useful information when removing the top principal components and PCA+GGM can still perform well.} \newup{In general, $\vb$ can hardly be closely aligned with the first few eigenvectors of $\Sigmab$. We verify this with simulations in Section \ref{sec:sim-pca+lvggm}. We randomly generate multiple confounders and notice that their eigenvecors are not closely aligned with the top eigenvectors of $\Sigmab$. Thus, PCA-based methods are effective to remove the confounders.}

\end{itemize}

\subsection{Comparison with LVGGM} \label{sec:compare_thm}
Now we compare LVGGM to PCA+GGM in more detail. 
%The convergence rate of PCA+GGM is of the order $O(\sqrt{\frac{p}{n}})$, the same as that of LVGGM. We know that $\sqrt{\frac{p}{n}}$ is the optimal rate for estimating the eigenvectors without further structural assumption \citep[Chapter 8]{wainwright2019high}, hence our result achieves the optimal rate.
We observe that PCA+GGM can be viewed as a supplement to LVGGM. The assumptions of PCA+GGM can be well satisfied when the assumptions of LVGGM cannot be satisfied. In (\ref{eqn:sigmaobs_proof}), now let $\vb$ be the $k$-th eigenvector of $\Omegab$ (thus the $(p-k+1)$-th eigenvector of $\Sigmab$), the Sherman-Morrison identity gives
\begin{equation}
    \label{eqn:compare_thm}
    {\Sigmab}^{-1}_{obs}=\Omegab-\frac{\lambda_k(\Omegab)^2}{\lambda_k(\Omegab)+(1/\sigma^2)}\vb \vb^\tT = \Omegab - \Lb_{\Omegab}.
\end{equation}
We can see that as $\sigma$ increases, ${\lambda_k(\Omegab)^2}/{(\lambda_k(\Omegab)+(1/\sigma^2))}$ increases. In the simulations in Section \ref{sec:simulation}, we observe that LVGGM performs poorly when $\vb$ is closely aligned with the first few eigenvectors of $\Omegab$ (thus the last few eigenvectors of $\Sigmab$). \new{One way to interpret why LVGGM does not work well under this setting is because the nuclear norm penalty in LVGGM will shrink large eigenvalues. Specifically, when $k$ is small and $\sigma^2$ is large, ${\lambda_k(\Omegab)^2}/{(\lambda_k(\Omegab)+(1/\sigma^2))}$ is large. Therefore, the nuclear norm regularization in LVGGM introduces larger bias. Additionally, when $k$ is small, $\vb$ is one of the top eigenvectors of $\Omegab$. We empirically observe that the top eigenvectors of $\Omegab$ can be coherent with standard basis.} \newup{This observation coincides with the identifiability issue of LVGGM, i.e., a sparse low-rank component worsens the performance of LVGGM.}

This observation is consistent with the conclusion in \cite{agarwal2012noisy}. They impose a spikiness condition, which is a weaker condition than the incoherence condition in \cite{chandrasekaran2012latent}. The spikiness condition requires that $\Vert \Lb_{\Omegab} \Vert_{\infty}$ is not too large. \eqref{eqn:compare_thm} shows that $\Lb_{\Omegab}$ tends to have a larger spectral norm when $\vb$ is aligned with the first few eigenvectors of $\Omegab$ and $\sigma$ is large, since in this case, ${\lambda_k(\Omegab)^2}/{(\lambda_k(\Omegab)+(1/\sigma^2))}$ is close to $\lambda_1(\Omegab)$. The large norm of $\Lb_{\Omegab}$ implies that the spikiness condition is not well satisfied, thus the error bound of LVGGM is large. Note, however, that the first few eigenvectors of $\Omegab$ are the last few eigenvectors of $\Sigmab$. Our analysis shows that the error bound of the estimate of PCA+GGM is small when $\vb$ is aligned with the first few eigenvectors of $\Omegab$ and $\sigma$ is large.

%In summary, when the norm of the low-rank confounding ($\Lb_{\Sigmab}$) is substantially large compared to the covariance matrix of interest $\Sigmab$ and its eigenvectors are not closely aligned with the first few eigenvectors of $\Sigmab$, we can consider using PCA+GGM.

\subsection{PCA+LVGGM}
\label{sec:pca+lvggm_thm}

In this section we discuss the PCA+LVGGM method briefly. We use the same formulation as (\ref{eqn:spluslomega_prime}) to (\ref{eqn:sum-of-low-rank-components}). We claim that PCA+LVGGM outperforms using PCA+GGM or LVGGM individually when $\Lb_{\Sigmab}$'s spectral norm is large compared to that of $\Lb_{\Omegab}^{'}$ and $\Sigmab$, $\Lb_{\Sigmab}$'s vectors are not aligned with the first few eigenvectors of $\Sigmab$, and the spectral norm of $\Lb_{\Omegab}^{'}$ is not significantly larger than that of $\Sigmab$. 
%We know that PCA+GGM can effectively remove $\Lb_{\Sigmab}$ if the norm of $\Lb_{\Sigmab}$ is larger than that of $\Sigmab$ and $\Lb_{\Omegab}^{'}$, and the eigenvectors of $\Lb_{\Sigmab}$ are not very closely aligned the first eigenvectors of $\Sigmab$. Therefore after using PC-correction to remove $\Lb_{\Sigmab}$, we obtain a good estimate of ${\Sigmab}^{'}$ in equation (\ref{eqn:covpluslomega_prime}). From \cite{chandrasekaran2012latent}, we know that with a good estimate of ${\Sigmab}^{'}$, LVGGM can effectively address $\Omegab$ and $\Lb_{\Omegab}$ in equation (\ref{eqn:spluslomega_prime}) if $\Lb_{\Omegab}$ is not sparse. Therefore we can obtain a good estimate of $\Omegab$ after applying PCA+LVGGM.
This is because based on Theorem \ref{thm1} and \ref{thm2}, PCA+GGM is effective only when the spectral norm of the low-rank confounding is larger than that of the signal. PC-correction, however, can only effectively remove $\Lb_{\Sigmab}$ but not $\Lb_{\Omegab}^{'}$ because the norm of $\Lb_{\Omegab}^{'}$ is not significantly larger than that of $\Sigmab$. In contrast, LVGGM can estimate $\Lb_{\Omegab}^{'}$ well, but not $\Lb_{\Sigmab}$ because it has a larger spectral norm and its eigenvectors might be aligned with the first few eigenvectors of $\Omegab$. 
%We found that in practice, it is not hard to satisfy the three conditions that can make PCA+LVGGM performs the best.

\subsection{Tuning parameter selection} \label{sec:tuning}

In both LVGGM and PCA+GGM there are crucial tuning parameters to select. For LVGGM, recall that $\lambda$ controls the sparsity of $\Omegab$ and $\gamma$ controls the rank of $\Lb_{\Omegab}$. \cite{chandrasekaran2012latent} argues that $\lambda$ should be proportional to $\sqrt{{p}/{n}}$, the rate in the convergence analysis, and choose $\gamma$ among a range of values that makes the graph structure of $\hat{\Omegab}$ stable \cite[see][for more detail]{chandrasekaran2011rank}. 

% In practice, if we have already selected the rank, we can try several $\gamma$ values to achieve that rank. Usually, the rank won't change very much when we change the other parameter $\lambda$ over some interval.

When using PCA+GGM, we need to determine the rank first (i.e., how many principal components should be removed). \cite{leek2007capturing, leek2012sva} suggest using the \texttt{sva} function from \texttt{Bioconductor}, which is based on parallel analysis \cite{horn1965rationale,buja1992remarks,lim2019determining}. Parallel analysis compares the eigenvalues of the sample
correlation matrix to the eigenvalues of a random
correlation matrix for which no factors are assumed. Given the number of principal components to remove, we can use model selection tools such as AIC, BIC or cross-validation to choose the sparsity parameter in Glasso. 
% Sometimes PC-correction removes too useful signal because too many pricipal components are removed. Thus, we may consider removing only half or one quarter of the number suggested by \texttt{sva}. 
One may also decide how many principal components to remove by considering the number of top eigenvalues of the observed covariance matrix (see Section \ref{sec:convergence} and Section \ref{sec:analysis_thm}). \new{Note that these rank selection approaches perform well when the low-rank confounding has large enough spectral norm compared to the norm of signal (more details on these conditions are discussed in \cite{leek2012sva, lim2019determining}). We will see in later sections that these conditions can be satisfied in many real-world applications. When the spectral norm of the latent confounder is small, approaches which do not account for confounding, such as Glasso and CLIME, are actually robust enough to perform well even when confounding exists. This is theoretically proved by \cite{ren2012discussion} and our simulations in next section also confirm this.}

% , if the observed covariance matrix is approximately low-rank, it may imply that large norm low-rank confounding exists. The number of large eigenvalues can be used as the rank. 

The PCA+LVGGM method has three tuning parameters: the rank of $\Lb_{\Sigmab}$, $\gamma$ and $\lambda$.  To start, we first look at eigenvalues or use the \texttt{sva} package to determine the total rank of the low-rank component, $\Lb_{\Sigmab} + \Lb^{'}_{\Omegab}$. \new{We think it is natural to determine the rank of confounder first because we will see in later applications, we can have some domain knowledge on the ranks of coufounders, e.g. in finance applications, some financial theory suggests the number of latent variables in the market.} We then need to partition the total rank between $\Lb_{\Sigmab}$ and $\Lb'_{\Omegab}$.  If we determine that rank$(\Lb_{\Sigmab} + \Lb'_{\Omegab}) = k$, we look for an eigengap in the first $k$  eigenvalues and allocate the largest $m < k$ eigenvalues for PC-removal. Our experiments in Section \ref{sec:application} show that domain knowledge can be used to motivate the number of components for PC-removal. After removing the principal components, we choose $\gamma$ in LVGGM so tha $\Lb'_{\Omegab}$ is approximately rank $k-m$ . We observe that when running LVGGM, the rank won't change for a range of $\lambda$ values when using a fixed $\gamma$. Thus, it suffices to fix $\gamma$ first to control the rank, then determine $\lambda$ to control the sparsity.

Practically, network estimation is often used to help exploratory data analysis and hypothesis generation. For these purposes, model selection methods such as AIC, BIC or cross-validation may tend to choose models that are too dense \cite{danaher2014joint}. This fact can also be observed by our experiments. Therefore, we recommend that model selection should be based on prior knowledge and practical purposes, such as network interpretability and stability, or identification of important edges with low false discovery rate \cite{meinshausen2010stability}. Thus, we recommend that the selection of tuning parameters should be driven by applications. For example, for biological applications, the model should be biologically plausible, sufficiently complex to include important information and sparse enough to be interpretable. In this context, a robustness analysis can be used to explore how edges change over a range of tuning parameters.

\section{Simulations} \label{sec:simulation}
{In this section, numerical experiments illustrate the utility of each sparse plus low rank method.  In Section \ref{sec:sim-part1} we illustrate the behavior of Glasso, LVGGM and PCA+GGM under different assumptions about rank-one confounding. In Section \ref{sec:sim-pca+lvggm}, we show the efficacy of PCA+LVGGM in a variety of simulation scenarios.  In all experiments, we set $p = 100$ and use the scale-free and random networks from \textsf{huge.generator} function in \textsf{R} package \textsf{huge} \cite{zhao2012huge}. \new{To generate random networks, each pair of off-diagonal elements are randomly set, while the graph is generated using B-A algorithm under scale-free structures \cite{albert2002statistical}.} Due to space limit, we only include results on the scale-free structure. }
%and provide random structure results in Appendix \ref{app2}.}

% \sout{In section \ref{sec:sim-part1}, we build intuition about the theoretical results from section \ref{sec:theory}: we use synthetic data to illustrate how the choice of a rank-1 confounder affects the relative performance of LVGGM and PCA+GGM. In section \ref{sec:sim-pca+lvggm}, we show the efficacy of PCA+LVGGM. When creating $\Omega$, we set $p = 100$ and use the scale-free and random schemes from \texttt{huge.generator} function in \texttt{R} package \texttt{huge}. We only include the results using the scale-free structure, the results using random structure are in Appendix \ref{app2}.}
% }

\subsection{The efficacy of LVGGM and PCA+GGM} \label{sec:sim-part1}
We compare the relative performance of PCA+GGM, LVGGM and Glasso in the presence of a rank-one confounder, $\Lb$. Guided by our analysis in Section \ref{sec:theory}, we show that the relationship between $\Lb$ and the eigenstructure of $\Sigmab$ determines the performance of these three methods. We first generate the data with
\begin{align*}
{\Xb}_{obs}^{(i)}&=\Xb^{(i)}+\Lb^{(i)}, ~~\Lb^{(i)} = {\sigma}\Vb \Zb^{(i)}, ~~i \in \{1,\cdots,n\},
\end{align*}
where $\Xb^{(i)} \in \mathbb{R}^p$ is normally distributed with mean zero and covariance matrix $\Sigmab$. $\Zb^{(i)} \in \mathbb{R}^r$, the low-rank confounder, follows a normal distribution with mean zero and  identity covariance matrix. $\Vb \in \mathbb{R}^{p \times r}$ is a non-random semi-orthogonal matrix satisfying $\Vb^\tT \Vb = \Ib$, and $\sigma \in \mathbb{R}$ represents the magnitude of the confounder. Without loss of generality, we assume that $\Xb^{(i)}$ and $\Zb^{(i)}$ are independent. We illustrate the performance of different methods under various choices for, $\Vb$, the eigencomponents of $\Lb$. 

We first set $r=1$, $p=100$ and $n=200$. The largest eigenvalue of $\Sigmab$ is around 5. We use $\vb_i$ to denote the $i$-th eigenvector of $\Sigmab$. When examining the effect of $\sigma$, we choose the $95$-th eigenvector of $\Sigmab$ as $\Vb$ to ensure that $\Vb$ is not closely aligned with the first few eigenvectors of $\Sigmab$. We then compare the cases with $\sigma^2 = 20$ and $3$. Next, we examine the effect of eigenvectors. We fix $\sigma^2$ as 20, and use the $i$-th eigenvector of $\Sigmab$ as $\Vb$, where $i \in \{1, 60, 95 \}$. Following previous notation, we use $\vb_1$, $\vb_{60}$ and $\vb_{95}$ as $\Vb$. 1 is chosen as the rank for PC-correction and LVGGM. We generate ROC curves \cite[9.2.5]{hastie2009elements} for each method based on  50 simulated samples and use the average to draw the ROC curves (Fig.~\ref{fig:sim-part1}). We truncate the ROC curves at FPR=0.2, since the estimates with large FPR are typically less useful in practice.

\begin{figure}[ht!]
    \centering
    \includegraphics[width=0.9\textwidth]{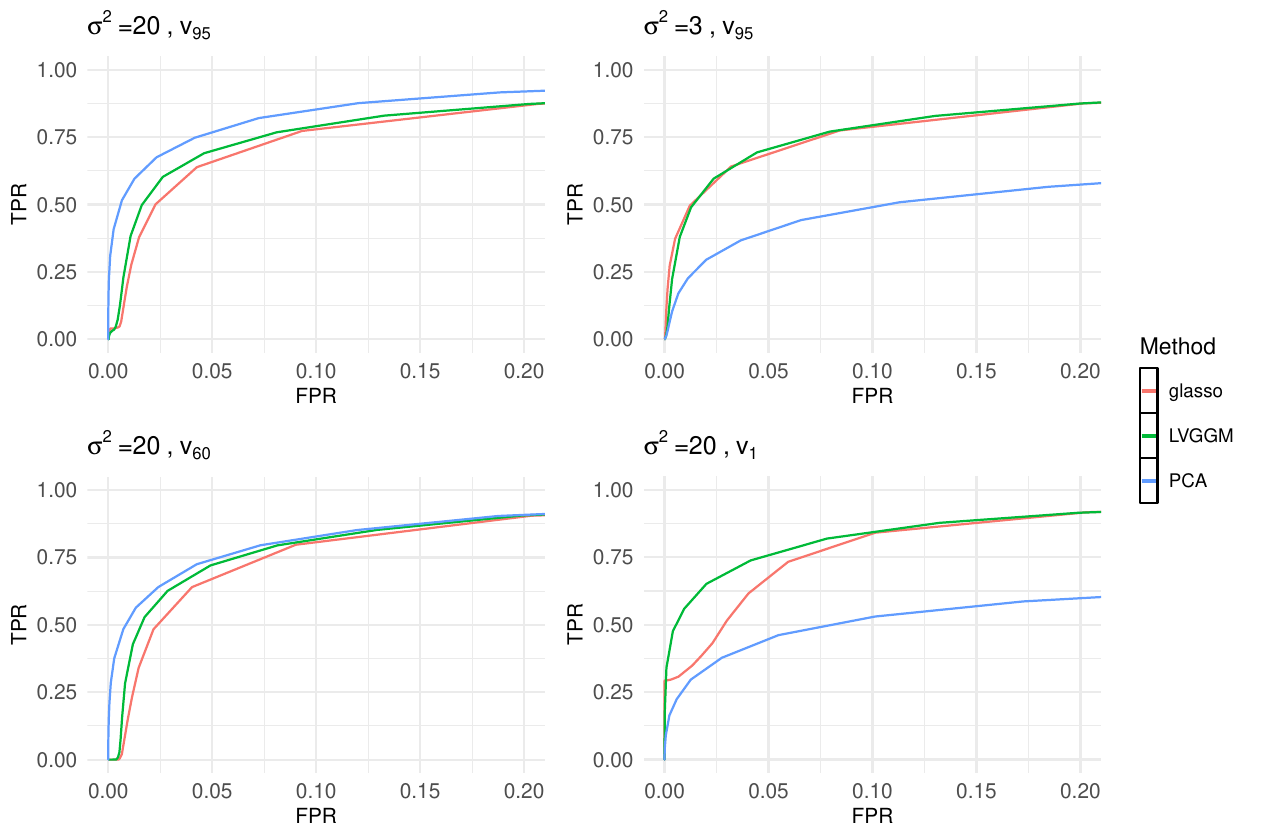}
    \caption{$n=200$. The low-rank component has rank 1. $\sigma^2$ is the magnitude of the low-rank component, $\vb_i$ is the $i$-th eigenvector of $\Sigmab$. The first row illustrates the effect of $\sigma$: we use the $95$-th eigenvector of $\Sigmab$ as $\Vb$, and set $\sigma^2$ to 20 and 3 from left to right. The second row illustrates the effect of $\Vb$: we fixed $\sigma^2$ as 20, and use the $60$-th and 
   first eigenvector of $\Sigmab$ as $\Vb$ from left to right. PCA+GGM works the best when $\sigma^2$ is large and $\Vb$ is not aligned with the first eigenvector of $\Sigmab$. LVGGM works the best when $\sigma^2$ is large and $\Vb$ is not aligned with the last eigenvector of $\Sigmab$. When $\sigma^2$ is small, Glasso works as well as the other two.}
    \label{fig:sim-part1}
\end{figure}

From Fig.~\ref{fig:sim-part1}, we observe that when the confounder has large norm and its eigenvectors are not closely aligned with the first few eigenvectors of $\Sigmab$, PCA+GGM performs better than LVGGM and Glasso. LVGGM preforms the best when the confounder has large norm and its eigenvectors are not aligned with the last few eigenvectors of $\Sigmab$ (also the first few eigenvectors of $\Sigmab^{-1}$). When the low-rank component does not have a large norm, Glasso also performs well. \new{This reaffirms the fact that Glasso can be robust enough to address the low-rank confounding with small norm.}
%Additional simulations for random graph structures and higher rank confounders, as well as different sample sizes, $n$, can be found in Appendix \ref{app2}.

\subsection{The efficacy of PCA+LVGGM} \label{sec:sim-pca+lvggm}
In this section, we use multiple examples to demonstrate the efficacy of the PCA+LVGGM. We introduce corruption of the signal with two low-rank confounders. The data is generated as follows:
\begin{align*}
    {\Xb}^{(i)}_{obs}=\Xb^{(i)}+\Vb_1 {\Db_1} \Zb_{1}^{(i)}+\Vb_2 {\Db_2} \Zb_{i}^{(i)}, \ i \in \{1,\cdots,n\},
\end{align*}
where $\Xb_i \in \mathbb{R}^p$ is normally distributed with  mean zero and covariance matrix $\Sigmab$. $\Zb_{1}^{(i)} \in \mathbb{R}^{d_1}$, corresponding to the first source of low-rank confounder, has a normal distribution with mean zero and  covariance matrix $\Ib_{d_1}$. $\Vb_1 \in \mathbb{R}^{p \times d_1}$ is a non-random, semi-orthogonal matrix satisfying $\Vb_1^\tT\Vb_1 = \Ib$, and $\Db_1 \in \mathbb{R}^{d_1 \times d_1}$ is a diagonal matrix, measuring the magnitude of the first confounder. Similarly, $\Zb_{2}^{(i)} \in \mathbb{R}^{d_2}$, corresponding to the second source of low-rank confounder, has normal distribution with mean zero and  covariance matrix $\Ib_{d_2}$. $\Vb_2 \in \mathbb{R}^{p \times d_2}$ is a semi-orthogonal matrix satisfying $\Vb_2^\tT\Vb_2 = \Ib$, and $\Db_2 \in \mathbb{R}^{d_2 \times d_2}$ is a diagonal matrix, measuring the magnitude of the second low-rank confounder. Without loss of generality, we assume that $\Xb^{(i)}$, $\Zb_{1}^{(i)}$ and $\Zb_{2}^{(i)}$ are three pairwise independent vectors. Hence the observed covariance matrix is
\begin{equation}
    \label{eqn:sigma_obs_sim}
    Cov({\Xb}_{obs})={{\Sigmab}_{obs}}=\Sigmab+\Vb_1\Db_1^2\Vb_1^\tT +\Vb_2\Db_2^2\Vb_2^\tT=\Sigmab+\Lb_1+\Lb_2.
\end{equation}

Our first simulation setup (case 1) shows an ideal case for PCA+LVGGM, meaning that PCA+LVGGM method performs much better than using PCA+GGM, LVGGM, or Glasso. Let $d_1= d_2 =3$. We set $p=100$ and $n=100$. In our first example, the columns of $\Vb_1$ and $\Vb_2$ come from the eigenvectors of $\Sigmab$. We expect that PC-correction removes $\Lb_2$, so we set the diagonal elements of $\Db_2^2$ all $50$, and use the last $3$ eigenvectors of $\Sigmab$ as $\Vb_2$. This can guarantee that PC-correction performs much better than LVGGM and Glasso when removing $\Lb_2$. Then we use LVGGM to estimate $\Lb_1$, so we need a moderately large magnitude. We set all diagonal elements of $\Db_1^2$ to $20$, and use the first $3$ eigenvectors of $\Sigmab$ as $\Vb_1$. This ensures that LVGGM performs better than PC-correction and Glasso when estimating $\Lb_1$. 

Using the \texttt{sva} package, we estimate the rank of $\Lb_1+\Lb_2$ to be 6. Then we look at the eigenvalues of the observed sample covariance matrix and we can see the first $3$ eigenvalues are much larger than the $4$-th to $6$-th eigenvalues (shown in the top row of Fig.~\ref{fig:sim-pca+lvggm}). We therefore allocate $3$ to PC-correction, and $6-3=3$ to LVGGM. We also try allocating 1 to PC-correction and 5 to LVGGM. Then we compare more approaches, including using PC-correction individually by removing only 3 principal components or 6 principal components, using LVGGM with rank 6 for the low-rank component as well as the uncorrected approach Glasso. We still use 50 datasets and draw the ROC curve for the averages with varying sparsity parameters $\lambda$. The ROC for the scale-free example is in the bottom row of Fig.~\ref{fig:sim-pca+lvggm}. We also include the AUC (area under the ROC curve) for each method. We compare PCA+LVGGM with rank 3 in PC-correction with other methods. For each data set, we calculate (AUC of PCA+LVGGM)/(AUC of one other method), then compute the sample mean and sample standard deviation of that ratio over 50 data sets to compare the average performance and the variance of different methods. The results for the scale-free graph are shown in the first column of Table \ref{tab:auc}. %and the results for the random graph structure are given in Appendix \ref{app2}. 
We can see that PCA+LVGGM with rank 3 for PC-correction and 3 for LVGGM do perform much better than other methods for both graph structures, indicating that if the assumptions are satisfied, our method and parameter tuning procedure are useful. 

Finally, we try setups that are more similar to real world data. We still use (\ref{eqn:sigma_obs_sim}) to generate the data and set $p=100$ and $n=100$. Differently from previous settings, we now use some randomly generated eigenvectors as columns of $\Vb_1$ and $\Vb_2$. We look at the distribution of eigenvalues of gene co-expression and stock return data covariance matrices, and try to make simulation settings similar to those examples. We run two setups - the first is called a large-magnitude case (case 2), with $\Db_1^2$ a diagonal matrix with diagonal elements $(7,6,6)$ and $\Db_2^2$ a diagonal matrix with diagonal elements $(20,10,10)$. The second setup is referred to a moderately large magnitude case (case 3), in which the low-rank component has the same eigenvectors as the large-magnitude case, but the elements of $\Db_1$ and $\Db_2$ become smaller, with diagonal elements of $\Db_1^2$ $ (3,3,3)$ and diagonal elements of $\Db_2^2$ $(10,8,6)$. 

Using the \texttt{sva} package, we estimate the rank of $\Lb_1 + \Lb_2$ to be 6 for both case 2 and case 3. We observe that the first 3 eigenvalues are larger than the rest, so we allocate 3 to PC-correction and use $6-3=3$ as the rank for the low-rank component for LVGGM. We also try allocating 1 to PC-correction and $5$ to LVGGM, using PC-correction by removing only 3 PC-components and $6$ PC-components, using LVGGM with rank $6$ and using Glasso. Again, we run over 50 datasets and include ROC curves and AUC tables. 
\begin{figure}[ht!]
    \centering
    \includegraphics[width=1\textwidth]{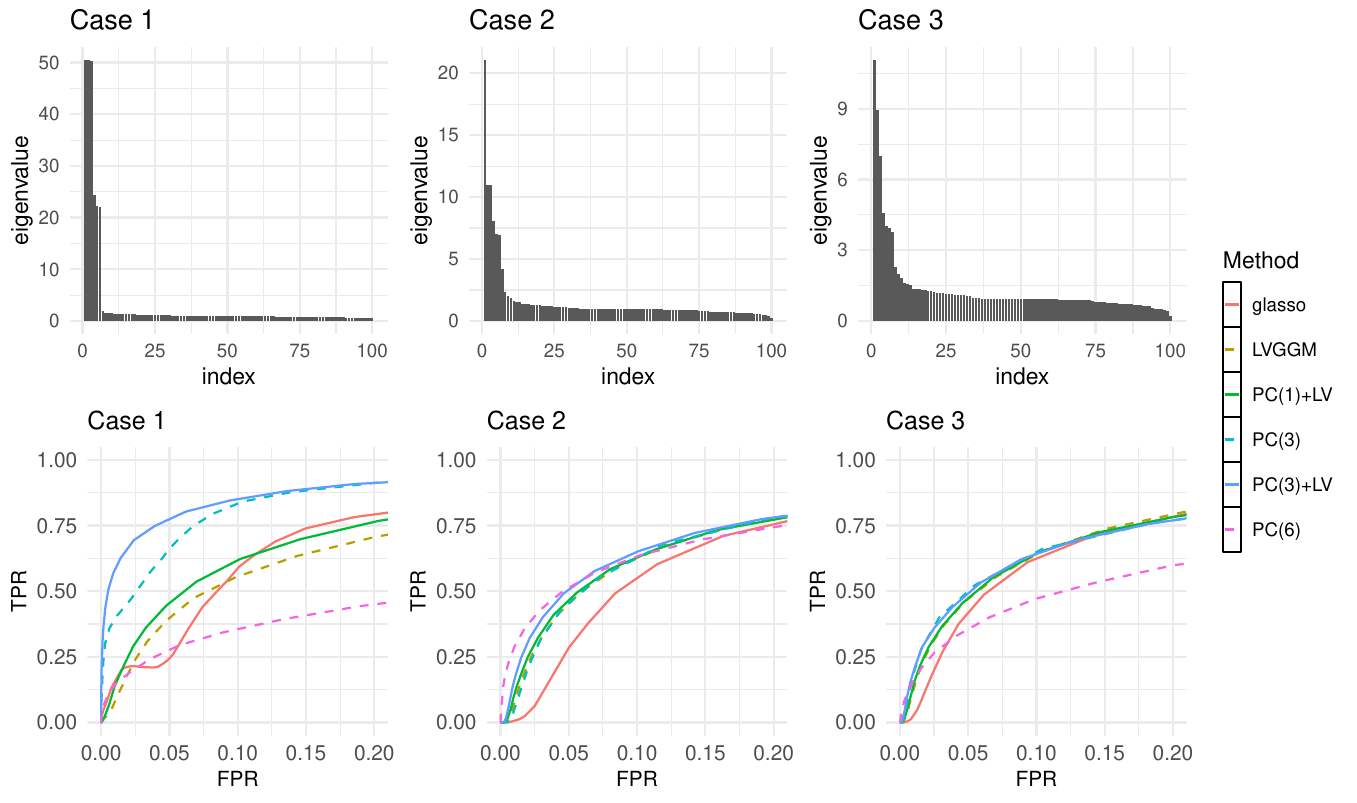}
    \caption{We use the scale-free structure when generating graphs. The first row shows the eigenvalues of ${\Sigmab}_{obs}$ under 3 setups, and the second row shows the corresponding ROC curves with different methods. PC(k) means that we use $k$ as the rank in PC-correction and PC(k)+LV means that we use PCA+LVGGM with $k$ as the rank for PC-correction.}
    \label{fig:sim-pca+lvggm}
\end{figure}

\begin{table}[ht!]

%\centering
\vskip-0.3cm\hrule

%\smallskip
\centering\small
 \begin{tabular}{c c c c} 
 %\hline
Method & Case 1 & Case 2 & Case 3\\ [0.7ex]
 %\hline
 PCA(3)+LVGGM & \textbf{1} & \textbf{1} & 1 \\
 %\hline
 Glasso & 1.58(0.073) & 1.25(0.080) & 1.07(0.048) \\
 %\hline
 LVGGM & 1.64(0.088) & 1.06(0.044) & 1.01(0.036) \\
 %\hline
 PCA(Full) & 2.46(0.22) & 1.01(0.12) & 1.36(0.14) \\
 %\hline
 PCA(3) & 1.08(0.017) & 1.05(0.027) & \textbf{0.99(0.018)} \\
  %\hline
 PCA(1)+LVGGM & 1.47(0.11) & 1.04(0.038) & 1.01(0.035) \\
 %\hline
\end{tabular}\\
\hrule
\label{tab:auc}
\caption{We use the scale-free structure when generating graphs. We compute the ratio of AUC between PCA+LVGGM with rank 3 in PC-correction and other methods, using PCA+LVGGM as the numerator. The table shows the sample mean and sample standard deviations of that ratio (in the parenthesis) over 50 data sets. In case 3, the magnitude of the confounding is not as large as other cases, so PC-correction with rank 3 has the best performance.}
\end{table}
From Fig.~\ref{fig:sim-pca+lvggm} and Table \ref{tab:auc}, we can see that other approaches considered hardly outperform PCA+LVGGM. Actually, using PCA+GGM or LVGGM can be viewed as a special case of the PCA+LVGMM methods. To see that, we can have LVGGM from PCA+LVGGM by allocating a rank of 0 to PC-correction. From the simulation and real data examples, we observe that using PCA+GGM with higher ranks often removes some useful information, resulting in more false negatives. On the other hand, if the effect of multiple confounders exists in the data that are not well represented by the first few principal components, using PCA+GGM alone might not be enough to remove the additional sources of noise. Note that LVGGM may not be enough to remove the confounding with large norm, leading to spurious connections between nodes. In this case, we would suggest PCA+LVGGM as a default setting and a starting point for problems with low-rank confounding. We can adjust different rank allocations based on the specific problems and goals of interest.

\section{Applications} \label{sec:application}

\subsection{Gene co-expression networks} \label{sec:gene}
Our first application is to reanalyze the gene co-expression networks originally analyzed by \cite{parsana2019addressing}. The goal of gene co-expression network analysis is to identify transcriptional patterns indicating functional and regulatory relationships between genes. In biology, it is of great interest to infer the structure of these networks; however, the construction of such networks from data is challenging, since the data is usually corrupted by technical and unwanted biological variability known to confound expression data. The influence of such artifacts can often introduce spurious correlation between genes; if we apply sparse precision matrix inference directly without addressing confounding, we may obtain a graph including many false positive edges. \cite{parsana2019addressing} uses PCA+GGM to estimate this network and shows that PC-correction can be an effective way to control the false discovery rate of the network. In practice, however, some effects of confounding may not be represented in the top few principal components. This motivates the more flexible PCA+LVGGM approach. The PC-correction effectively removes high variance confounding, and then LVGGM subsequently accounts for any remaining low-rank confounding. We consider gene expression data from $3$ diverse tissues: blood, lung and tibial nerve, with sample sizes between 300 to 400 each. 1000 genes are chosen from each tissue. More detail about the source of the data and pre-processing steps are introduced in Appendix \ref{app-gene}.

\begin{figure}[ht!]
    \centering
    \includegraphics[width=1\textwidth]{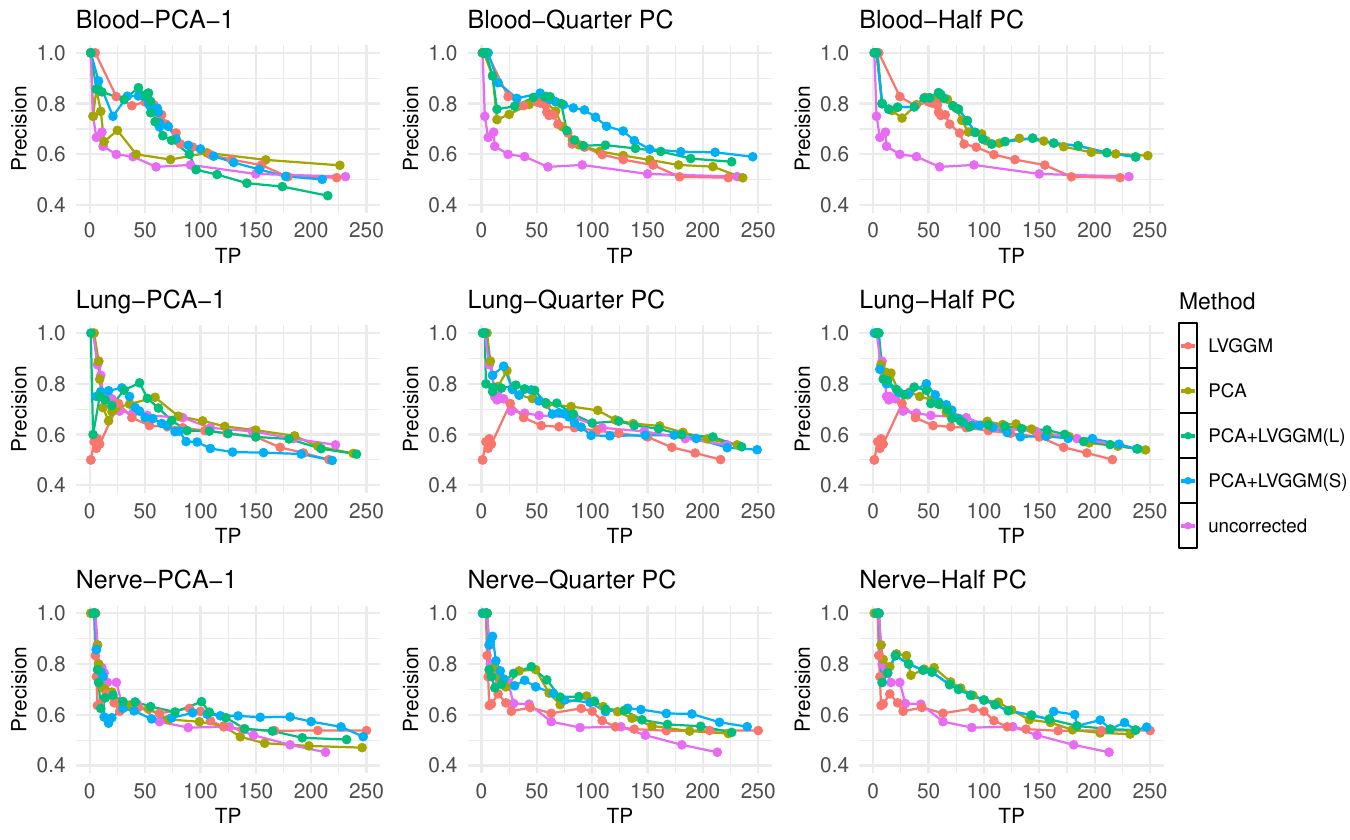}
    \caption{Precision-recall plots for gene expression data. TP represents number of true positives. PCA+LVGGM(L) means larger $\gamma$ in LVGGM and PCA+LVGGM(S) means small $\gamma$ in LVGGM. We can see that PCA+LVGGM performs the best or equivalently well compared to other approaches for almost all 3 tissues.}
    \label{fig:gene_1}
\end{figure}
We observe that all of the sample covariance matrices are approximately low-rank by looking at the eigenvalues of the covariance matrices of genes, indicating the potential existence of high variance confounding. %\new{The plots of the distributions of eigenvalues are shown in Appendix \ref{sec:eig-application}.} 
Then we use \texttt{sva} package to estimate the rank for PC-correction and call this the full \texttt{sva} rank correction. \cite{parsana2019addressing} suggests that the rank estimated by \texttt{sva} might be so large that some useful network signal is removed. To reduce the effect of over-correction, we apply the PC-correction with half and one quarter of the \texttt{sva} rank, which we refer to as half \texttt{sva} rank correction and quarter \texttt{sva} rank correction, respectively. For many tissues, the first eigenvalue is much larger than the rest, this motivates us to try rank-1  PC-correction. We include the results with half \texttt{sva} rank, quarter \texttt{sva} rank and rank 1 PC-corrections in Fig.~\ref{fig:gene_1}. After running the above PC-corrections to remove high-variance confounding, we run LVGGM as an additional step to further estimate and remove the low-rank noise with moderate variance. We use two different values as the $\gamma$ parameters in LVGGM. Larger $\gamma$ leads to removing lower-rank confounding and smaller $\gamma$ leads to remove higher-rank confounding. We show the results for both choices of $\gamma$. We use different $\lambda$ to control sparsity of the estimated graph. {\new{Specifically, following \cite{parsana2019addressing}, we use 50 values of $\lambda$ between 0.3 and 1.}} We draw Fig.~\ref{fig:gene_1} similar to the precision recall plot. The y-axis represents the precision (True Positives/(True Positives + False Positives)), and the x-axis is the number of true positives. We can see that PCA+LVGGM can yield better or equivalently good results compared to other methods, indicating that it can be useful to run LVGGM after the PC-correction when estimating gene co-expression networks.
% Since PCA+LVGGM is more general, the PCA+GGM method and LVGGM are two special cases, and we can always start with PCA+LVGGM. 

\subsection{Stock return data} \label{sec:stock}

In finance, the Capital Asset Pricing Model (CAPM) states that there is a widespread market factor which dominates the movement of all stock prices. Empirical evidence for the market trend can be found in the first principal component of the stock data, which is dense and has approximately equal loadings across all stocks (Fig.~\ref{fig:stocklowrank}, left).  In fact, the first few eigenvalues of the stock correlation matrix are significantly larger than the rest \cite{fama2004capital}, which suggests that only a few latent factors are \new{mainly} driving stock correlations.

In this section, we posit that the conditional dependence structure after accounting for these latent effects is more likely to reflect direct relationships between companies aside from the market and, perhaps, sector trends. Our interest is in recovering the undirected graphical model (conditional dependence) structure between stock returns after controlling for potential low rank confounders.

% \sout{In this application, we model the dependency structure of the monthly return of companies from the S\&P $100$ index from 2008 to 2019 \footnote{Data scraped from https://finance.yahoo.com/}.  Market beta theory \citep{fama2004capital} suggests that there exists a widespread market trend with a similar effect on the movement of all stock prices. This trend is empirically evident from the first principal component of the stock data, which is dense and has approximately equal loadings across all stocks (Figure \ref{fig:stocklowrank}).  In fact, the first few eigenvalues of the stock correlation matrix are significantly larger than the rest, which suggests that a few latent factors are driving stock correlations.  In particular, due to the overall market trends and sector effects \citep{fama2004capital}, stocks may appear to have significant partial correlations even if there are no direct interactions or business partnerships between the companies.  Our interest is recovering the undirected graphical model (conditional dependence) structure between stock returns \emph{after} controlling for the overall market trend and sector effects.  After controlling for these latent effects, we posit that the remaining conditional dependence is more likely to be reflect direct relationships between companies.  We analyze stock data from 49 companies in 6 sectors \citep{hayden2015canonical}: technology (10 companies), finance (11), energy (7), health (8), capital goods (7) and non-cyclical stocks (6).}

We compare networks inferred by PCA+LVGGM, PCA+GGM, LVGGM and Glasso by analyzing monthly returns of component stocks in S\&P $100$ index between 2008 and 2019 \cite{hayden2015canonical}. The 49 chosen companies are in 6 sectors: technology (10 companies), finance (11), energy (7), health (8), capital goods (7) and non-cyclical stocks (6). For PCA+GGM, we remove the first eigenvector which corresponds to the overall market trend.  For the other latent variable methods we use the \texttt{sva} package to identify a plausible rank. For PCA+LVGGM, we remove the first principal component corresponding to the overall market trend and use LVGGM to estimate remaining latent confounders and the graph.  Fig.~\ref{fig:stock-1} shows the networks obtained by each approach.

\begin{figure}[t]
\begin{minipage}[b]{0.5\linewidth}
  \centering
  \centerline{\includegraphics[width=5.75cm]{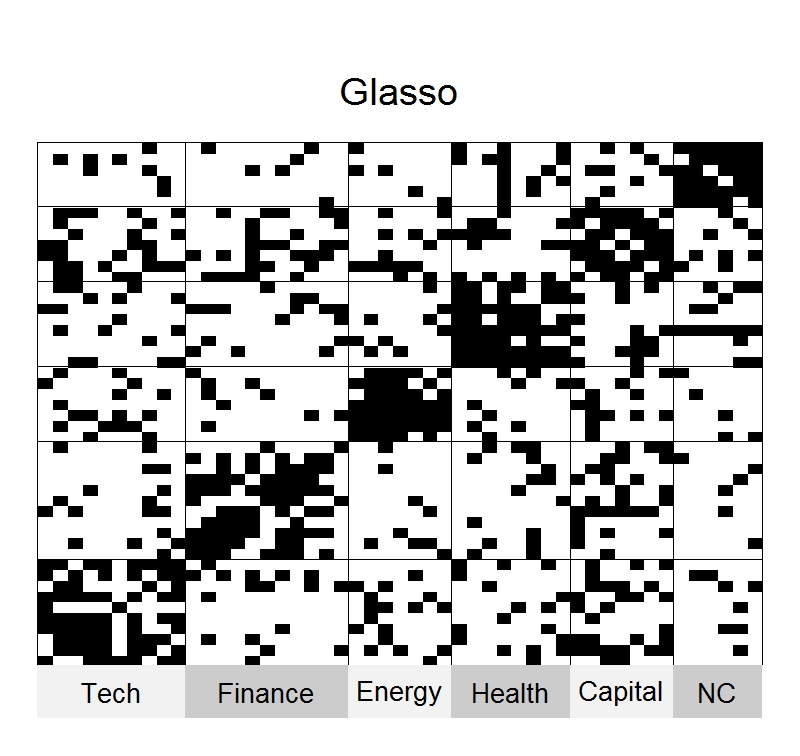}}
 %\vspace{0.05cm}
  %\centerline{(a)}\medskip
\end{minipage}
%\hspace{0.1cm}
\hfill
\begin{minipage}[b]{0.5\linewidth}
  \centering
  \centerline{\includegraphics[width=5.75cm]{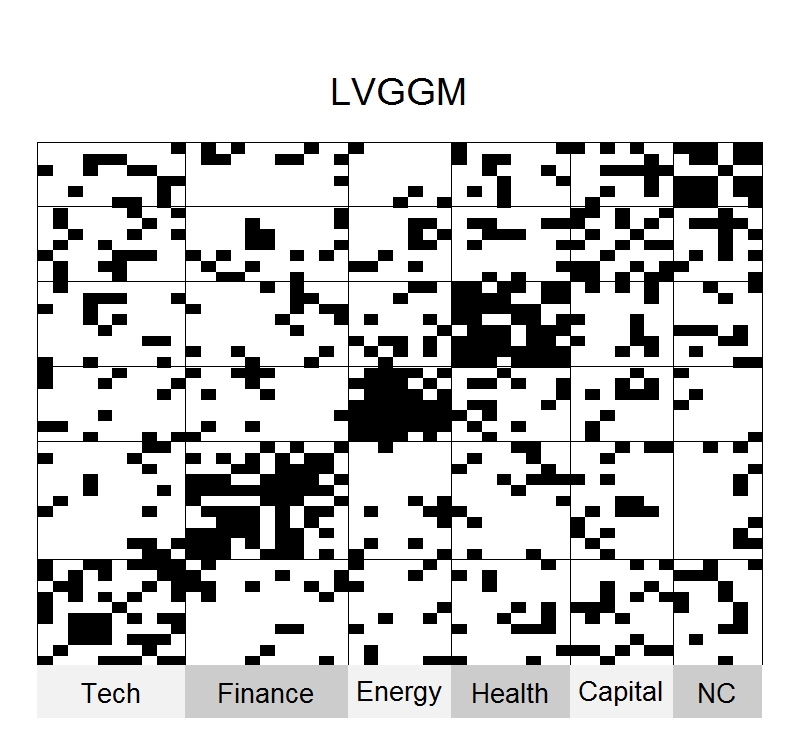}}
 %\vspace{0.05cm}
  %\centerline{(b) }\medskip
\end{minipage}
\begin{minipage}[b]{0.5\linewidth}
  \centering
  \centerline{\includegraphics[width=5.75cm]{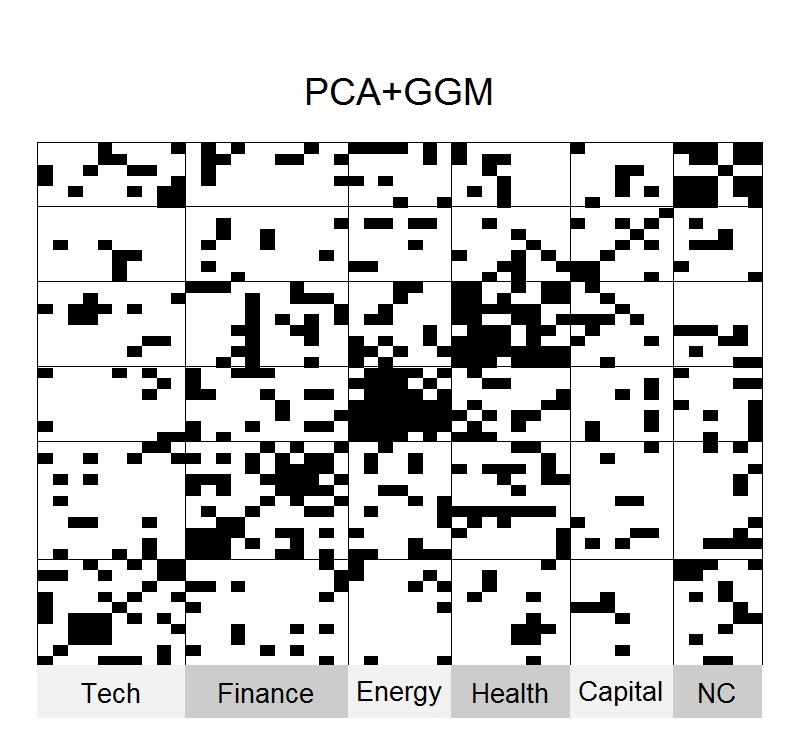}}
 %\vspace{0.05cm}
  %\centerline{(a)}\medskip
\end{minipage}
%\hspace{0.1cm}
\hfill
\begin{minipage}[b]{0.5\linewidth}
  \centering
  \centerline{\includegraphics[width=5.75cm]{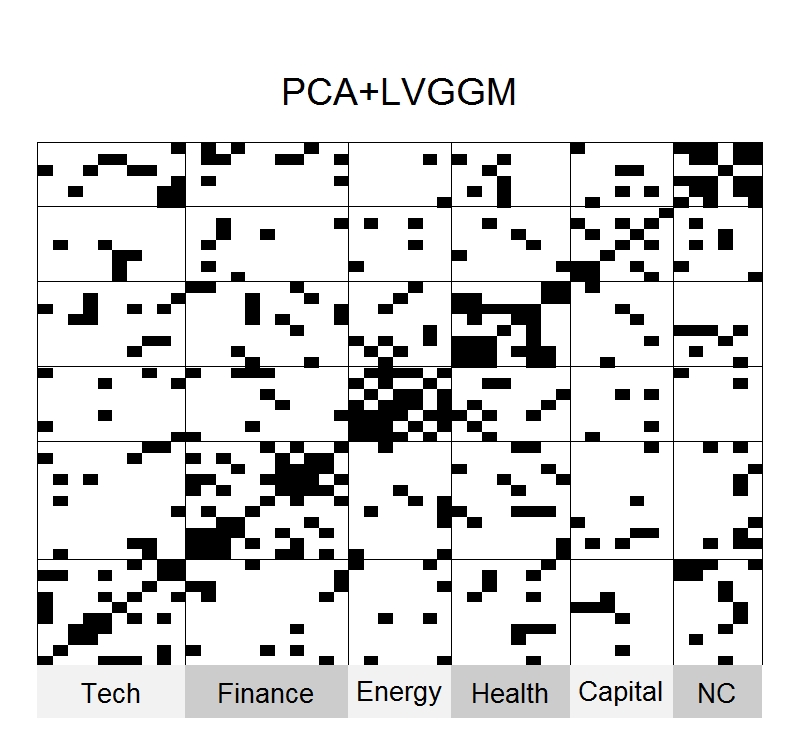}}
 %\vspace{0.05cm}
  %\centerline{(b) }\medskip
\end{minipage}
    \caption{Stock connections between 2008 and 2019 learned by different methods. The following sectors are included: Tech, Finance, Energy, Health, Capital goods and Non-cyclical (NC) from left to right.}
    \label{fig:stock-1}
\end{figure}

%We then look at the KL-divergence between the distribution with sample covariance matrix and the distribution with our estimates. The KL-divergence is defined as :
%\begin{equation}
  %  \label{eqn:kldivergence}
   % \frac{1}{2}[-\log(\det({\Sigma}_{obs}\hat{{\Omega}}_{obs}))+\text{tr}({\Sigma}_{obs}\hat{{\Omega}}_{obs})].
%\end{equation}
%When the sample size $n$ is large, we use the empirical covariance matrix $\hat{\Sigma}_{obs}$ as $\Sigma_{obs}$. The KL-divergence obtained from Glasso, LVGGM, PCA+GGM and PCA+LVGGM are 62.32, 59.39, 61.82 and 61.3 respectively. Note that when using PCA-based methods, after removing the PCs, we use a rank-deficient matrix to estimate a full-rank matrix, thus leading to higher bias. Even though, we can see Glasso still has the highest KL-divergence, indicating that Glasso is not as suitable as other methods.

For each method, the sparsity-inducing tuning parameter was chosen to minimize negative log-likelihood using a 6-fold cross-validation procedure, and the number of low rank components are chosen manually.  \new{Specifically, in cross-validation, we use negative log-likelihood to measure the out-of-sample error and choose the parameters that minimize the average out-of sample error over 6 validation sets. }%The average out-of-sample error measured by negative log-likelihood with the optimal tuning parameters for Glasso, PCA+GGM, LVGGM and PCA+LVGGM are 26.68, 27.16, 26.82 and 27.17 respectively. The minimum out-of-sample error measured by negative log-likelihood with the optimal tuning parameter for Glasso, PCA+GGM, LVGGM and PCA+LVGGM are 16.44, 15.83, 15.84 and 15.99 respectively. 
We observe that when using LVGGM, allocating rank 1 or 2 to the low-rank component won't make the estimates very different from Glasso, while allocating ranks higher than 6 to the low-rank component leads to higher out-of-sample error, so 5, the rank picked by sva, is among the best choices. For PCA-based methods, removing more than 1 principal components leads to higher out-of-sample error. As expected, the Glasso result is denser than the networks learned with sparse plus low rank methodology with PCA+LVGGM yielding the sparsest network.% This fact suggests that many connections between stocks may be due to the strong market trend or sector effect. 

% We argue that PCA+LVGGM is the most appropriate method for this application by comparing the low-rank matrices learned by each method. 

For LVGGM, we note that the method effectively controls for sector effect but is less effective in controlling for the effect of the overall market trend.  
Let $\hat{\Sigmab}_{obs}$ be the empirical observed covariance matrix and $\hat{\vb}_i$ be its $i$-th eigenvector. We have the following observations: first,  the first principal component is closely aligned with the overall market trend, because the absolute value of the inner product between the first eigenvector of $\hat{\Sigmab}_{obs}$ and the normalized ``all ones'' vector is 0.98. Second, the observed empirical covariance matrix has an approximately low-rank structure, because the first eigenvalue of $\hat{\Sigmab}_{obs}$ is $18.25$ and the second is $3.5$ and all other eigenvalues are close or smaller than 1. Third, LVGGM does not capture the full effect of the market trend. Let ${\hat{\Lb}_{\Omegab}}^{'}$ be the estimate of ${{\Lb}_{\Omegab}}^{'}$ in \eqref{eqn:covpluslomega_prime}. When we apply LVGGM on $\hat{\Sigmab}_{obs}$, the inner product between the first eigenvector of ${\hat{\Lb}_{\Omegab}}^{'}$ and $\hat{\vb}_1$ is close to $1$ but the first eigenvalue of ${\hat{\Lb}_{\Omegab}}^{'}$ is only $0.55$, much smaller than the first eigenvalue of $\hat{\Sigmab}_{obs}$.
%the estimate of $\Sigmab_{obs}^{-1}$ is 
%\begin{align*}
%   \hat{\Omegab} - \hat{\Lb}_{\Omegab},
%\end{align*}
%then applying Sherman-Morrison identity on the equation above gives the estimate of $\Sigmab_{obs}$:
%\begin{align*}
%     \hat{\Sigmab} + %{\hat{\Lb}_{\Omegab}}^{'},
%\end{align*}
%where $\hat{\Omegab}^{-1} = \hat{\Sigmab}$. the estimate of $\Sigmab_{obs}$ is:
%\begin{align*}
%     \hat{\Sigmab} + %{\hat{\Lb}_{\Omegab}}^{'} + %\hat{\Lb}_{\Sigmab}.
%\end{align*}

We argue that PCA+LVGGM is the most appropriate method for this application because it appropriately controls for both market and sector effects. Let $\hat{\Lb}_{\Sigmab}$ and ${\hat{\Lb}_{\Omegab}}^{'}$ be the estimates of the low-rank components defined in \eqref{eqn:sum-of-low-rank-components}. For PCA+LVGGM, we remove ${\Lb}_{\Sigmab}$ by removing the first eigencomponent of ${\Sigmab}_{obs}$, then run LVGGM to estimate ${{\Lb}_{\Omegab}}$ and $\Omegab$. We claim that PCA+LVGGM can remove the confounding effect fully in the market trend direction, as well as the remaining confounding effect in other directions. To see that, first, $\hat{\vb}_1$ is removed in PC-correction. Second, the inner product between the first eigenvector of ${\hat{\Lb}_{\Omegab}}^{'}$ and $\hat{\vb}_2$, the second eigenvector of $\hat{\Sigmab}_{obs}$, is $0.99$. The first eigenvalue of ${\hat{\Lb}_{\Omegab}}^{'}$ is $0.4$ and the second eigenvalue of $\hat{\Sigmab}_{obs}$ is $3.5$. This shows that when applying LVGGM, only part of the information in the direction of $\hat{\vb}_2$ has been removed. We know that the direction of $\hat{\vb}_1$ reflects the market trend, but $\hat{\vb}_2$ might include both true graph information and some latent confounding effect, hence using LVGGM might be a good choice for capturing the confounding effect in the direction of $\hat{\vb}_2$. Overall PCA+LVGGM, therefore, might be a better choice than LVGGM and the PCA-based method.

Fig.~\ref{fig:stocklowrank} shows heat maps of $\hat{\Lb}_{\Sigmab}$ obtained via PCA and $\hat{\Lb}_{\Omegab}^{'}$ obtained with LVGGM (rank $5$). As expected, the elements of $\hat{\Lb}_{\Sigmab}$ are roughly equal in magnitude, reflect the market trend and the large first eigenvalue of $\hat{\Sigmab}_{obs}$. In contrast, $\hat{\Lb}_{\Omegab}^{'}$ shows a block-diagonal structure and its elements have smaller magnitudes, which suggests that LVGGM does not adequately account for the overall the market trend. On the other hand, the block diagonal structure of $\hat{\Lb}_{\Omegab}^{'}$  reflects inferred sector effects. PCA+GGM is most effective at reducing confounding from overall market trends and LVGGM is more effective at accounting for remaining confounding, such as the sector effect. Therefore, PCA+LVGGM, which combines the benefits of PCA and LVGGM is arguably the most appropriate choice for addressing the latent confounding in this context.

%\begin{figure}[ht!]
%    \centering
%    \includegraphics[width=0.75\textwidth]{draft_plots/fig-stocklvggm02.png}
%    \caption{$\hat{L}_{\Omega}^{'}$ obtained from LVGGM.}
%    \label{fig:stock-lowranklvggm}
%\end{figure}

%\begin{figure}[ht!]
%    \centering
%    \includegraphics[width=0.75\textwidth]{draft_plots/fig-stockpca01.png}
%    \caption{$\hat{L}_{\Sigma}$ obtained from PCA+GGM.}
%    \label{fig:stock-lowrankpca}
%\end{figure}

%\begin{figure}[ht!]
%    \centering
%    \includegraphics[width=0.75\textwidth]{draft_plots/fig-stockpcalv01.png}
%    \caption{$\hat{L}_\Omega^{'}$ obtained from PCA+LVGGM.}
%    \label{fig:stock-lowrankpcalv}
%\end{figure}

%\begin{figure}[ht!]
%    \centering
%    \includegraphics[width=0.75\textwidth]{draft_plots/fig-stocklowlvggm.pdf}
%    \caption{LVGGM pdf plot.}
%    \label{fig:stock-lowrankpcalvpdf02}
%\end{figure}

\begin{figure}[t]
\begin{minipage}[b]{0.5\linewidth}
  \centering
  \centerline{\includegraphics[width=7cm]{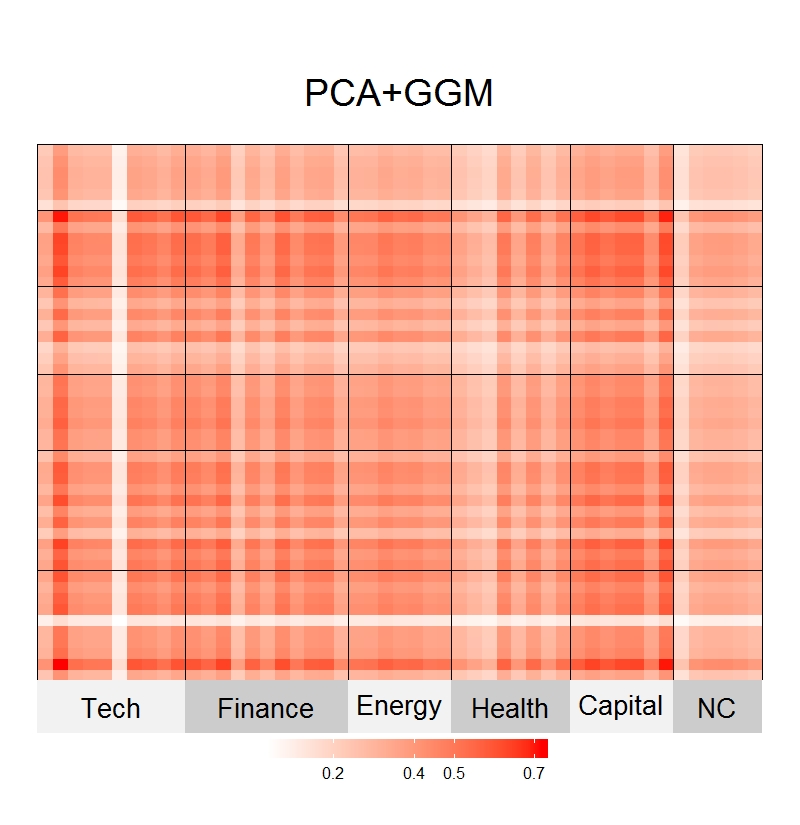}}
 \vspace{0.05cm}
  \centerline{(a)}\medskip
\end{minipage}
%\hspace{0.1cm}
%\hfill
\begin{minipage}[b]{0.5\linewidth}
  \centering
  \centerline{\includegraphics[width=7cm]{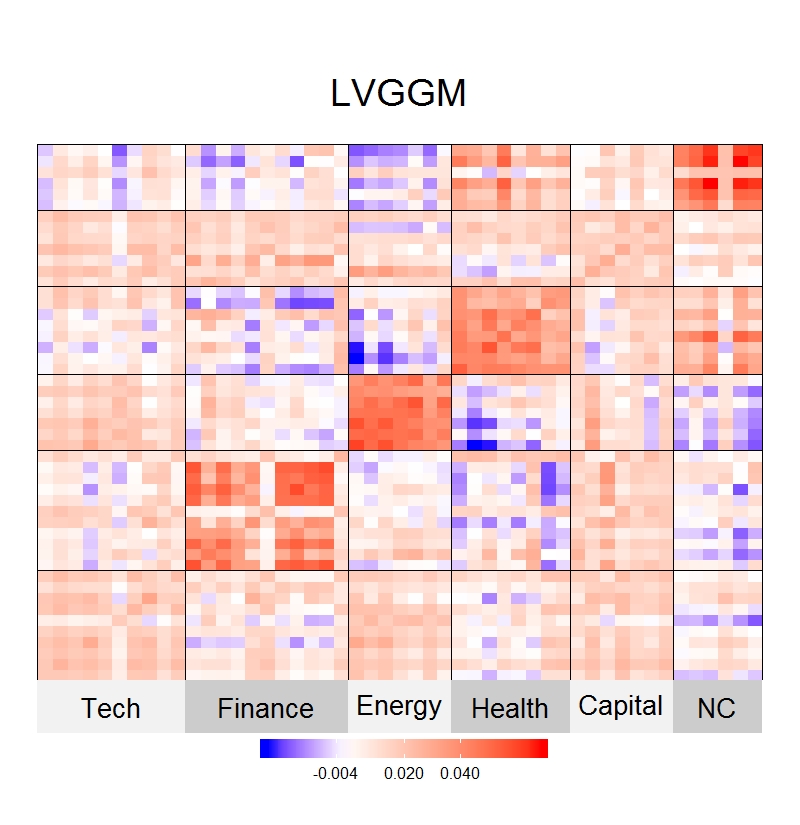}}
 \vspace{0.05cm}
  \centerline{(b) }\medskip
\end{minipage}
%\hfill
\caption{(a) Rank one approximation to $\hat{\Lb}_{\Sigmab}$ obtained with PCA. (b) $\hat{\Lb}_{\Omegab}^{'}$ obtained with LVGGM with rank $5$.  The rank-one approximation to $\hat{\Lb}_{\Sigmab}$ is close to a constant matrix.  In contrast, $\hat{\Lb}_{\Omegab}^{'}$ reflects sector effects but does not reflect the strong effect due to overall market trends.}
\label{fig:stocklowrank}
\end{figure}

\section{Conclusion} \label{sec6}
We have studied the problem of estimating the graph structure under Gaussian graphical models when the data is corrupted by latent confounders. We compare two popular methods PCA+GGM and LVGGM. One of our contributions is to show the connection and difference between these two approaches, both theoretically and empirically. Based on that, we propose a new method, PCA+LVGGM. The effectiveness of this method is supported by theoretical analysis, simulations and applications. Actually, our analysis provides guidance on when to use which approach to estimate GMM in the presence of latent confounders. We believe that this guidance can help researchers in many fields, such as finance and biology, when there exist problems of graph estimation with confounders.

There are several future directions. First, we can extend the current framework to other distributions, such as the transelliptical distributions \cite{liu2012transelliptical} and the Ising model \cite{ravikumar2010high,nussbaum2019ising}. Secondly, we can consider more structures - for example, we are interested in what will happen if the principal components of the observed data are sparse. \new{Given the distributions of eigenvalues or eigenvectors of $\Sigmab$, we can provide more precise bounds. }\newup{Deriving sharper bounds to more general settings is a useful extension. }\new{We can also extend LVGGM based on our observations in this work. For instance, it is interesting to consider replacing the nuclear norm penalty in LVGGM with some unbiased regularization such as SCAD \cite{fan2001variable} and MCP \cite{zhang2010nearly} to avoid shrinking too much on large eigenvalues.} Other future works include applying our current methods and analysis in more applications, such as the functional magnetic resonance imaging (fMRI) in neuroscience.

\section*{Acknowledgments}
The authors thank Christos Thrampoulidis and Megan Elcheikhali for the helpful discussion. The authors thank anonymous reviewers for their valuable feedback.

%\section*{Declarations of Interest}
%None.

\begin{appendices}\label{app}
\section{Proofs}\label{appl}

%\subsection{\new{Proof outline and auxiliary lemmas}}
%\label{pf-outline}
\begin{proof}[\textbf{\upshape Proof outline:}]\label{pf-outline}
To prove Theorems \ref{thm1} and \ref{thm1new}, we first write $\hat{\Sigmab}-\Sigmab$ as below:
\begin{align}
    \hat{\Sigmab}-\Sigmab &= (\hat{{\Sigmab}}_{obs}-\hat{\lambda}_1 \hat{\thetab}_1 \hat{\thetab}_1^\tT) - ({\Sigmab}_{obs} - \lambda_1 \thetab_1 \thetab_1^\tT +  \lambda_k \thetab_1 \thetab_1^\tT) = (\hat{{\Sigmab}}_{obs}-{\Sigmab}_{obs})+(\lambda_1 \thetab_1 \thetab_1^\tT - \hat{\lambda}_1 \hat{\thetab}_1 \hat{\thetab}_1^\tT) - \lambda_k \thetab_1 \thetab_1^\tT, \label{pf-decomp}
\end{align}
where $\lambda_k$ is the $k$-th eigenvalue of ${\Sigmab}_{obs}$, and $\thetab_k$ is the $k$-th eigenvector of ${\Sigmab}_{obs}$. To bound $\hat{\Sigmab}-\Sigmab$, we need to bound the norms of ${\Sigmab}_{obs}-\hat{{\Sigmab}}_{obs}$, $\lambda_1-\hat{\lambda}_1$ and $\thetab_1-\hat{\thetab}_1$, and the lemmas below show these bounds.
\end{proof}
\begin{lemma}
\label{lm1}
Assuming that ${\Sigmab}^{-1}_{obs}$ satisfies the maximum and minimum eigenvalue condition in (\ref{eqn:sigmaeig_proof}), then
\begin{align*}
    P(\Vert \hat{{\Sigmab}}_{obs}-{\Sigmab}_{obs}\Vert_\infty \ge t) \le C_1 p^{-1}, \ \ \ t=C_2\sqrt{\frac{\lnn p}{n}},
\end{align*}
where $C_1$ and $C_2$ depends on the eigenvalue bound $M$ in \eqref{eqn:sigmaeig_proof}.
\end{lemma}
\begin{proof}[\textbf{\upshape Proof:}] 
The proof follows \cite[Lemma A.3]{bickel2008regularized}.
\end{proof}

\new{The next two lemmas provide bounds for $|\lambda_1-\hat{\lambda}_1|$.}

\begin{lemma}
\label{lm2}
Under the assumptions of Theorem \ref{thm1}, we have
\begin{align*}
    |\lambda_1 - \hat{\lambda}_1| \le C_1\lambda_1 \sqrt{\frac{p}{n}},
\end{align*}
with probability at least $1-2e^{-p/C_2}$ for some constants $C_i$'s $> 1$.
\end{lemma}
\begin{proof}[\textbf{\upshape Proof:}] 
By Weyl's lemma \cite{horn2012matrix}
\begin{align*}
    \max_{j=1,...,p}|\lambda_j({\Sigmab}_{obs})-\lambda_j(\hat{{\Sigmab}}_{obs})| \le \Vert \hat{\Sigmab}_{obs}-\Sigmab_{obs} \Vert_2.
\end{align*}
The bound on $ \Vert \hat{\Sigmab}_{obs}-\Sigmab_{obs} \Vert_2$ is then obtained from \cite[Theorem 6.5]{wainwright2019high}.
\end{proof}
\new{Following \cite[Theorem 5]{koltchinskii2017concentration}, a tighter bound can be obtained using the effective rank $r(\Sigmab_{obs})$ defined in \eqref{eqn:eff_rank}. }
\new{\begin{lemma}
\label{lm2new}
Under the assumptions of Theorem \ref{thm1new}, we have
\begin{align*}
    \newup{|\lambda_1 - \hat{\lambda}_1| \le \Vert \hat{\Sigmab}_{obs}-\Sigmab_{obs} \Vert_2 \le C_1(\sqrt{\frac{p\lambda_1}{n}}\vee\frac{p}{n}),}
\end{align*}
with probability at least $1-e^{-p/C_2}$ for some constants $C_i$'s $> 1$.
\end{lemma} }
\begin{proof}[\textbf{\upshape Proof:}] 
\newup{Theorem 5 in \cite{koltchinskii2017concentration} shows that
\begin{align*}
    |\lambda_1 - \hat{\lambda}_1| \le \Vert \hat{\Sigmab}_{obs}-\Sigmab_{obs} \Vert_2 \le C\lambda_1 (\sqrt{\frac{ r(\Sigmaobs)}{n}}\vee\frac{r(\Sigmaobs)}{n}).
\end{align*}
The bound above is at the same order as $\sqrt{{(p\lambda_1)}/{n}}\vee({p}/{n})$ because $r(\Sigmaobs) \asymp ({p}/{\lambda_1})$.}
\end{proof}

\new{Then the two lemmas below provide bounds for $\| \thetab_1-\hat{\thetab}_1 \|_\infty$.}

\begin{lemma}[adapted from Wainwright 2019, Corollary 8.7]
\label{lm3}
Under the assumptions of Theorem \ref{thm1}, suppose $n \ge p$ and $\Vert \Sigmab \Vert_2 \sqrt{{(\nu + 1)}/{\nu^2}}\sqrt{{{p}/{n}}}\le {1}/{128}$, then
\begin{align*}
    \Vert \thetab_1 - \hat{\thetab}_1 \Vert_2 \le C_1 \sqrt{\frac{\nu + 1}{\nu^2}}\sqrt{\frac{p}{n}},
\end{align*}
with probability at least $1- C_2e^{-p/C_3}$ for some $\Cis$, where $\nu=\lambda_1({\Sigmab}_{obs})-\lambda_2({\Sigmab}_{obs})$ .
\end{lemma}

\new{The lemma below shows a tighter bound with a large eigengap $\nu=\lambda_1({\Sigmab}_{obs})-\lambda_2({\Sigmab}_{obs})$ following \cite[Section 3.1]{fan2018eigenvector}.}

\new{\begin{lemma}
\label{lm3new}
Under the assumptions of Theorem \ref{thm1new}, 
\begin{align*}
    \newup{\Vert \thetab_1 - \hat{\thetab}_1 \Vert_\infty \asymp O_p(\frac{p \vee \lambda_1}{\nu p}\sqrt{\frac{\lnn p}{n}}) \asymp O_p(\frac{\sqrt{p}}{\lambda_1}\sqrt{\frac{\lnn p}{n}}),}
\end{align*}
with probability at least $1- C/p$ for some constant $C$.
\end{lemma}  }

\new{Before moving to the proofs of Theorem \ref{thm1} and \ref{thm1new}, we first show an upper bound for $\| \hat{\Sigmab} - \Sigmab \|_\infty$ using \eqref{pf-decomp}.
\begin{align}
    \Vert \hat{\Sigmab}-\Sigmab \Vert_{\infty} \le \Vert {\Sigmab}_{obs}-\hat{{\Sigmab}}_{obs} \Vert_{\infty}+\Vert \lambda_1 \thetab_1 \thetab_1^\tT - \hat{\lambda}_1 \hat{\thetab}_1 \hat{\thetab}_1^\tT \Vert_{\infty} + \Vert\lambda_k \thetab_1 \thetab_1^\tT \Vert_{\infty}. \label{pf-decomp02}
\end{align}
The term $\Vert \lambda_1 \thetab_1 \thetab_1^\tT - \hat{\lambda}_1 \hat{\thetab}_1 \hat{\thetab}_1^\tT \Vert_{\infty}$ can be expressed as:
\begin{align}
    \Vert \lambda_1 \thetab_1 \thetab_1^\tT - \hat{\lambda}_1 \hat{\thetab}_1 \hat{\thetab}_1^\tT \Vert_{\infty} = &\| \lambda_1 \thetab_1 \thetab_1^\tT - {\lambda}_1 {\thetab}_1 \hat{\thetab}_1^\tT + {\lambda}_1 {\thetab}_1 \hat{\thetab}_1^\tT - \hat{\lambda}_1 {\thetab}_1 \hat{\thetab}_1^\tT + \hat{\lambda}_1 {\thetab}_1 \hat{\thetab}_1^\tT - \hat{\lambda}_1 \hat{\thetab}_1 \hat{\thetab}_1^\tT \|_{\infty} \notag \\
    \le & \| \lambda_1 \thetab_1 \thetab_1^\tT - {\lambda}_1 {\thetab}_1 \hat{\thetab}_1^\tT \|_{\infty} + \| {\lambda}_1 {\thetab}_1 \hat{\thetab}_1^\tT - \hat{\lambda}_1 {\thetab}_1 \hat{\thetab}_1^\tT\|_{\infty} + \|\hat{\lambda}_1 {\thetab}_1 \hat{\thetab}_1^\tT - \hat{\lambda}_1 \hat{\thetab}_1 \hat{\thetab}_1^\tT \|_{\infty} \notag\\
    \le & |\lambda_1|\|\thetab_1\|_\infty\Vert \thetab_1 - \hat{\thetab}_1 \Vert_\infty + |\lambda_1 - \hat{\lambda}_1|\|\thetab_1\|_\infty\|\hat{\thetab}_1\|_\infty + |\hat{\lambda}_1|\|\hat{\thetab}_1\|_\infty\Vert \thetab_1 - \hat{\thetab}_1 \Vert_\infty. \label{pf-decomp03}
\end{align}
We can then use the bounds for $|\lambda_1 - \hat{\lambda}_1|$ and $\Vert \thetab_1 - \hat{\thetab}_1 \Vert_\infty$ in previous lemmas to bound $\Vert \lambda_1 \thetab_1 \thetab_1^T - \hat{\lambda}_1 \hat{\thetab}_1 \hat{\thetab}_1^\tT \Vert_{\infty}$.
}

%\subsection{Proof of Theorems \ref{thm1} and \ref{thm1new}}
\begin{proof}[\textbf{\upshape Proofs of Theorems \ref{thm1} and \ref{thm1new}:}]\label{pf-pfthm1}
\new{Now we are ready to prove Theorem \ref{thm1}. We first plug in the bounds in Lemmas \ref{lm1}, \ref{lm2} and \ref{lm3} to \eqref{pf-decomp03}. Since $\thetab_1$ is the eigenvector of a matrix, $\| \thetab_1 \|_\infty \le 1$. Combining these two results completes the proof.}

\new{To prove Theorem \ref{thm1new}, we need to plug in the bounds in Lemmas \ref{lm1}, \ref{lm2new} and \ref{lm3new} to \eqref{pf-decomp03}. \newup{Given Lemmas \ref{lm2new} and \ref{lm3new}, we can bound the second term in \eqref{pf-decomp03},}
\begin{align*}
    \newup{|\lambda_1 - \hat{\lambda}_1|\|\thetab_1\|_\infty\|\hat{\thetab}_1\|_\infty = O_p(\sqrt{\frac{\lambda_1}{np}}\vee\frac{1}{n}) = O_p(\sqrt{\frac{\lnn p}{n}}).}
\end{align*}
\newup{Similarly, the first and third terms in \eqref{pf-decomp03} have the convergence rate of}
\begin{align*}
    \newup{O_p(\lambda_1\Vert \thetab_1 \Vert_\infty (\frac{\sqrt{p}}{\lambda_1}\sqrt{\frac{\lnn p}{n}})) = O_p(\sqrt{\frac{\lnn p}{n}}).}
\end{align*}
\newup{Then with incoherent $\| \thetab_1 \|_\infty = O(1/\sqrt{p})$, the last term in \eqref{pf-decomp02} $\Vert\lambda_k \thetab_1 \thetab_1^\tT \Vert_{\infty}$ has the convergence rate of $O_p(1/p)$.}}
\end{proof}

%\subsection{Proof of Theorem \ref{thm2}}
\begin{proof}[\textbf{\upshape Proof of Theorem \ref{thm2}:}]\label{pf-thm2}
The proof follows the proof of \cite[Theorem 6]{cai2011constrained}. First we know,
\begin{align*}
    \Vert \hat{\Sigmab}\Omegab-\Ib \Vert_{\infty}=\Vert (\hat{\Sigmab}-\Sigmab)\Omegab \Vert_{\infty} \le \Vert \hat{\Sigmab}-\Sigmab \Vert_{\infty}\Vert \Omegab \Vert_{L_1}.
\end{align*}
Then we have,
\begin{align*}
    \Vert \hat{\Sigmab}(\Omegab-\hat{\Omegab}_1)\Vert_{\infty}&=\Vert \hat{\Sigmab}\Omegab - \Ib +\Ib- \hat{\Sigmab}\hat{\Omegab}_1\Vert_{\infty}
    \le \Vert  \hat{\Sigmab}\Omegab - \Ib\Vert_{\infty}+\Vert \Ib- \hat{\Sigmab}\hat{\Omegab}_1\Vert_{\infty}+\Vert \hat{\Sigmab} - \Sigmab \Vert_{\infty} \Vert \Omegab \Vert_{L_1} + \lambda_n.
\end{align*}
We know,
\begin{align*}
    \Vert \Omegab-\hat{\Omegab}_1\Vert_{\infty}=\Vert \Omegab \Sigmab(\Omegab-\hat{\Omegab}_1) \Vert_{\infty} \le \Vert \Sigmab(\Omegab-\hat{\Omegab}_1) \Vert_{\infty} \Vert \Omegab \Vert_{L_1}.
\end{align*}
To bound the terms above, we need,
\begin{align*}
    \Vert \Sigmab (\Omegab-\hat{\Omegab}_1)\Vert_{\infty} \le \Vert \hat{\Sigmab}(\Omegab-\hat{\Omegab}_1)\Vert_{\infty} + \Vert \hat{\Sigmab} - \Sigmab \Vert_{\infty}\Vert \Omegab \Vert_{L_1}.
\end{align*}
We know $\Vert \Omegab \Vert_{L_1} \le M_0$ from (\ref{eqn:class_clime}) and combining the relations above with the result of Theorem \ref{thm1} or Theorem \ref{thm1new}, with the choice of $\lambda_n$ specified in Theorem \ref{thm2}, we can obtain the bound for $\hat{\Omegab}_1$. The bound of the same order can be obtained for $\hat{\Omegab}$, the symmetric version of $\hat{\Omegab}_1$.
\end{proof}

\section{Generalization of Section \ref{sec:theory}} \label{app12}
The analysis in Section \ref{sec:theory} assumes that the low-rank confounder is independent of $\Xb$ and the eigenvector of the covariance of the low-rank confounding is one of the eigenvectors of $\Sigmab$, the covariance of $\Xb$. Those two assumptions can be extended to the more general setups. In equation (\ref{eqn:xplusl_proof}), when $\Xb$ and $\Zb$ are not independent, the convariance matrix for ${\Xb}_{obs}$ becomes
\begin{align*}
    {\Sigmab}_{obs}= \Sigmab + {\sigma}{Cov}(\Xb,\Zb)\vb^\tT + {\sigma}\vb {Cov}(\Xb,\Zb)^\tT + \sigma^2 \vb \vb^\tT,
\end{align*}
where $Cov(\Xb,\Zb)$ is a $p$-dimensional column vector. We can see that ${\sigma}Cov(\Xb,\Zb)\vb^\tT + {\sigma}\vb Cov(\Xb,\Zb)^\tT + \sigma^2 \vb \vb^\tT$ has rank at most 3, hence ${\Sigmab}_{obs}$ can still be expressed as the sum of $\Sigmab$ and a low-rank matrix. \new{Here, to ensure that the confounding can be identified in PCA-based approach, we assume that both $\sigma$ and $\sigma^2$ are large compared to the eigenvalues of $\Sigmab$.} Then, our analysis in Section \ref{sec:theory} can still be applied here, but the eigenvectors of the low-rank matrix are not necessarily the eigenvectors of $\Sigmab$.

\section{Eigenvalues of sparse graphs}
\label{app-eigen}
\begin{figure}[ht!]
    \centering
    \includegraphics[width=0.85\textwidth]{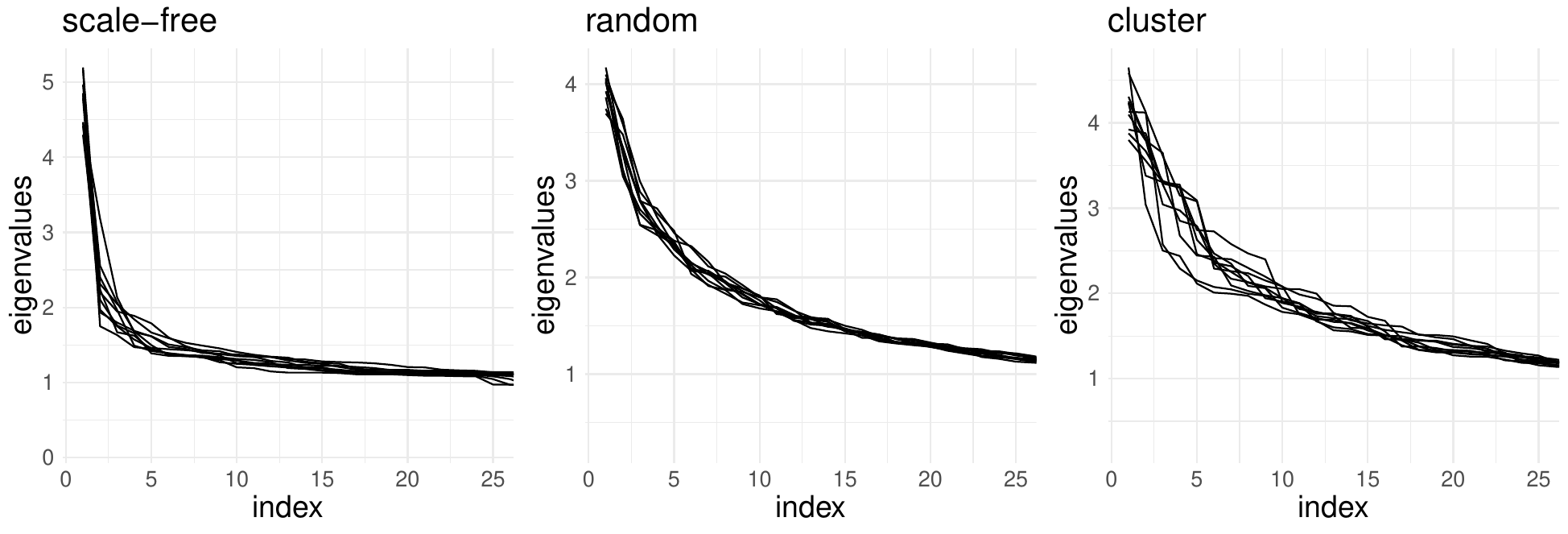}
    \caption{The distribution of the top 25 eigenvalues created by scale-free, random and cluster graphs using \texttt{huge} package with $p = 100$ and $n = 10000$.}
    \label{fig:eigenvalues01}
\end{figure}
\new{Fig.~\ref{fig:eigenvalues01} shows the first 25 eigenvalues of sparse graphs. We use \textsf{huge} package \cite{zhao2012huge} to generate the sparse graphs with three different structures: scale-free, random and clustered. We set $p = 100, n = 10000$ and assume default for all other parameters (see \cite{zhao2012huge} for more detail). We generate 10 realizations for each graph structure and show the distribution of the first 25 eigenvalues of $\Sigmab$ in Fig.~\ref{fig:eigenvalues01}. We notice that the top eigenvalues are typically larger than the rest, especially for the scale-free graphs.}

\section{Gene co-expression networks data}\label{app-gene}
Now we briefly introduce the data and the pre-processing procedure of gene co-expression networks in Section \ref{sec:gene}. More detail can be found from \cite{parsana2019addressing}. We use the RNA-Seq data from the Genotype-Tissue Expression (GTEx) project v6p release. %\footnote{https://www.gtexportal.org/home/}. 
We consider three diverse tissues with sample sizes between 300 to 400 each: blood, lung and tibial nerve. We first filter the non-overlapping protein genes and perform a log transformation with base $2$ to scale the data following \cite[Appendix 2.4]{parsana2019addressing}. Since the underlying true network structure is generally unknown, we obtain the interaction information from some canonical pathway databases including KEGG, Biocarta, Reactome and Pathway Interaction Database. To make better use of those information, we pick 1000 high-variance genes which are included in all these databases, thus $p=1000$ in this example.

%\section{\new{Eigenvalues Distributions of Applications}}
%\label{sec:eig-application}
%\new{Figure~\ref{fig:eig-application} shows the eigenvalues distributions of the covariance matrices in the genetic and stock return application. We can see that they are all very spiky. The very large top eigenvalues implies the existence of latent confounding.}

%\begin{figure}[ht!]
%\centering
%\begin{minipage}[b]{0.74\linewidth}
%  \centering
%  \centerline{\includegraphics[width=12cm]{draft_plots/bio_egvl.pdf}}
%  \centerline{(a)}
%\end{minipage}
%
%\begin{minipage}[b]{0.25\linewidth}
%  \centering
%  \centerline{\includegraphics[width=4.5cm]{draft_plots/stockeigen01.pdf}}
%  \centerline{(b)}
%\end{minipage}
%\caption{(a) shows the eigenvalue distributions of the covariance matrices of 3 tissues analyzed in section \ref{sec:gene}. (b) shows the eigenvalue distributions of the covariance matrices of stock returns analyzed in section \ref{sec:gene}.}
%\label{fig:eig-application}
%\end{figure}

\end{appendices}

\bibliography{references}

\end{document}